\documentclass[10pt, twocolumn]{IEEEtran}
\usepackage{algorithm,algorithmic,amsbsy,amsmath,amssymb,epsfig,bbm,mathrsfs,fancyhdr,fancyvrb,url,color,multirow}
\usepackage{lipsum}
\usepackage{subcaption}
\usepackage[utf8]{inputenc}
\usepackage[english]{babel}

\usepackage[normalem]{ulem}
\usepackage{cite}

\newcommand{\beq}{\begin{equation}}
\newcommand{\eeq}{\end{equation}}
\newcommand{\beqn}{\begin{eqnarray}}
\newcommand{\eeqn}{\end{eqnarray}}
\usepackage{multirow}

\newtheorem{lemma}{Lemma}

\newtheorem{corollary}{Corollary}
\newenvironment{proof}[1][Proof:]{\begin{trivlist}
\item[\hskip \labelsep {\bfseries #1}]}{\end{trivlist}}

\newcommand{\qed}{\nobreak \ifvmode \relax \else
      \ifdim\lastskip<1.5em \hskip-\lastskip
      \hskip1.5em plus0em minus0.5em \fi \nobreak
      \vrule height0.75em width0.5em depth0.25em\fi}

\begin{document}
\title{{\color{black} Mobility-Aware Modeling and Analysis of Dense Cellular Networks with C-plane/U-plane Split Architecture  
 }}

\author{Hazem Ibrahim, Hesham ElSawy, Uyen T. Nguyen, and Mohamed-Slim Alouini
\IEEEcompsocitemizethanks{\IEEEcompsocthanksitem Hazem Ibrahim and Uyen T. Nguyen are with the Department
of Electrical Engineering and Computer Science, York University, Toronto, Ontario, M3J 1P3 Canada (e-mail:{hibrahim,utn}@cse.yorku.ca).
\IEEEcompsocthanksitem Hesham ElSawy and Mohamed-Slim Alouini are with King Abdullah University of Science and Technology
(KAUST), Al-Khawarizmi Applied Math Building, Thuwal 23955-6900, Makkah Province, Kingdom of Saudi Arabia (e-mail: {hesham.elsawy,slim.alouini}@kaust.edu.sa).}
}




\maketitle

\begin{abstract}

The unrelenting increase in the population of mobile users  and their traffic demands drive cellular network operators to densify their network infrastructure.  Network densification shrinks the footprint of base stations (BSs) and reduces the number of users associated with each BS, leading to an improved spatial frequency reuse and spectral efficiency, and thus, higher network capacity. However, the densification gain come at the expense of higher handover rates and network control overhead. Hence, users mobility can diminish or even nullifies the foreseen densification gain. In this context, splitting the control plane (C-plane) and user plane (U-plane) is proposed as a potential solution to harvest densification gain with reduced cost in terms of handover rate and network control overhead.  In this article, we use stochastic geometry to develop a tractable mobility-aware model for a two-tier downlink cellular network with ultra-dense small cells and C-plane/U-plane split architecture. The developed model is then used to quantify the effect of mobility on the foreseen densification gain with and without C-plane/U-plane split. To this end, we shed light on the handover problem in dense cellular environments, show scenarios where the network fails to support certain mobility profiles, and obtain network design insights.
\end{abstract}


%

\section{Introduction}\label{sec:introduction_report}

\IEEEPARstart{T}{he} fifth generation (5G) of cellular networks is challenged to enhance  users' experience, support new services, and satisfy the ever-increasing mobile user population and their traffic demands. Compared to the state-of-the-art 4G cellular systems,  5G networks are expected to achieve thousandfold capacity improvement with at least hundredfold increase in the peak data rate {\color{black}and one order of magnitude delay reduction}~\cite{Andrews_5G}. Researchers in both academia and industry almost agree that network densification, via base station deployment, is among the key solutions to achieve this ambitious performance goal \cite{Andrews_5G}.  Therefore, it is expected that cellular network operators will significantly densify their networks infrastructures to fulfill the 5G performance requirements. In this case, network densification via deployments of small base stations (SBSs) is preferred over deployments of macro base stations (MBSs) due to lower cost and faster deployment.

Deploying more SBSs within the same geographical region reduces the footprint of each BS, and thus, decreases the number of users served by each BS.  Reduced BS footprints shorten user-to-serving-BS distances and improve the spatial frequency reuse.   Therefore, network densification is foreseen to improve spatial spectral efficiency and thus network capacity. However, narrowing BS footprints leads to higher handover rates and control overhead per unit area.  The increased handover rate imposes a major challenge that may negate the foreseen densification gain if conventional network operation is preserved. In extreme cases, where high mobility exists in urban areas (e.g., monorails in city downtowns or the Shinkansen network of high-speed railway in Tokyo), a densely deployed cellular network may fail to support very fast moving users due to excessive handover rates. Particularly, the network cannot support mobile users with a cell dwell time that is comparable or less than the  handover delay. Consequently, the undesirable effect of narrowing the BSs footprints requires solutions that reduce handover rate and  control overhead in order to harvest the foreseen network densification gain.

Decoupling control plane (C-plane) and user plane (U-plane) for cellular networks, under a cloud radio access network (C-RAN) umbrella,  is proposed as a potential solution to reduce handover rate and control burden \cite{split2012phantom}. Cellular network architecture with C-plane/U-plane (CP/UP) split  is also referred to as ``{\em Lean Carrier}'' for LTE \cite{hoymann2013lean}. Fig.~\ref{split_and_no_split} illustrates {\color{black} cellular network architecture with CP/UP split}. In this architecture, user devices can receive data packets from a nearby SBS while being controlled via  a farther MBS.  {\color{black}It is shown in \cite{hoymann2013lean, split2012phantom} that implementing the control plane at the macro cell level and the data plane at the small cell level incurs less control overhead compared to the conventional architecture (i.e., both C-plane and U-plane are jointly served from each BS). The CP/UP split architecture imposes less control overhead because the cell specific control signals/channels for SBSs, which identify each SBS, are not broadcast.\footnote{Examples of cell specific control signals/channels are primary/secondary synchronization signals [PSS/SSS], cell-specific reference signals [CRS], master information blocks [MIB] and system information blocks [SIB] (see \cite{split2012phantom, hoymann2013lean} for details).} Consequently, the SBSs become transparent to the users and the MBSs take charge of managing the radio resource control (RRC) procedures between mobile devices and  SBSs, such as session establishment and release. In the CP/UP split network, the SBSs are referred to as {\em phantom BSs}  because their identities are hidden from the users.\footnote{The abbreviation SBS in this article refers to both a small BS and a phantom BS.}


\begin{figure}[t]
\begin{center}
\scalebox{0.38}[0.38]{\includegraphics{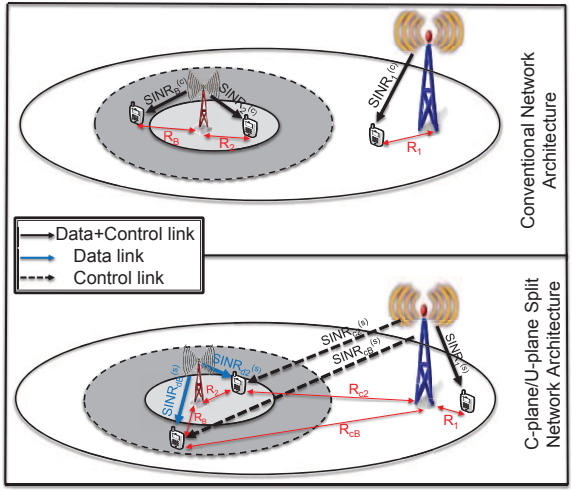}}
\end{center}
\caption{Conventional vs. CP/UP split network architecture: There are three types of links in the conventional network with three corresponding SINR values: ${\rm SINR}^{(C)}_1$ for macrocell users, ${\rm SINR}^{(C)}_2$ for non-biased users, and ${\rm SINR}^{(C)}_B$ for biased users. There are five types of links in the CP/UP split network with five corresponding SINR values: ${\rm SINR}^{(s)}_1$ for macrocell users, ${\rm SINR}^{(s)}_{d2}$ for non-biased users' data, ${\rm SINR}^{(s)}_{c2}$ for non-biased users' control, ${\rm SINR}^{(s)}_{dB}$ for biased users' data and ${\rm SINR}^{(s)}_{cB}$ for biased users' control and service distances.}
\vspace{-.2cm}
 \label{split_and_no_split}
\end{figure}

In addition to reducing the control overhead, the CP/UP split architecture can also be exploited to mitigate handover delays in dense cellular environments. Since the MBSs are in charge of the control signaling for the phantom cells including SBS selection, the MBSs can act as handover anchors and mange the handovers between underlying SBSs. In this case, the core network is only informed about inter-MBSs handovers. Compared to the conventional network architecture which informs the core network about MBSs and SBSs handovers. Hence, the CP/UP split architecture can significantly reduce handover delay by only reporting the less frequent inter-MBSs handovers, thanks to the larger coverage of macro cells. It is ought to be mentioned that the relative performance between the conventional  and CP/UP split architectures highly depends on the availability of the direct X2 interfaces between the BSs. This is because the X2 interface also enables core network transparent handover procedure. However, the X2 interface does not provide signaling overhead reduction as in the CP/UP split case.

In this article, we use stochastic geometry to develop a tractable mobility-aware model that characterizes the performance of cellular networks with and without CP/UP split. In particular, we model downlink transmission in two-tier cellular networks with flexible cell association, in which the model takes into account the impact of the handover rate and control overhead on users throughput. Tractable expressions for per-user throughput in terms of the BSs intensity, users velocity, and handover delay are obtained to study the effect of mobility on throughput in dense cellular environments, in which the performances of conventional and CP/UP split architectures are compared. \textcolor{black}{To this end, we shed light on the handover delay problem in dense cellular environments and show the potential delay mitigation via the CP/UP split architecture. The developed model is also used to quantify the expected performance gain for the CP/UP split architecture, obtain design insights, and discuss the performance limits of the conventional and CP/UP split architectures. To the best of our knowledge, this article is the first to develop a theoretical and tractable mobility-aware modeling paradigm to study the handover problem in dense cellular environments and evaluate the performance of the CP/UP split network architecture. Based on the developed model, potential scenarios where  CP/UP split is essential to support user mobility are highlighted and the feasibility of CP/UP split is also discussed}. 


The remainder of the article is organized as follows. In Section \ref{related_work}, we provide an overview of the related work. In Section \ref{System_Model_and_Assumption}, we provide the system model and assumptions. Section~\ref{pectrum_Allocation_Control_Burden} presents the conventional and CP/UP split transmission rate models. Section \ref{Performance_Analysis} characterizes the coverage probability and spectral efficiencies of the conventional and CP/UP split architectures. Section \ref{mobility_analysis} presents mobility analysis and evaluate the handover costs incurred by mobile users.  We validate the proposed model and discuss numerical results in Section \ref{validation_and_result}. Finally, Section \ref{conc_future} concludes the paper and outlines our future work.

\section{Related Work}\label{related_work}



Since modern cellular networks exhibit random topologies rather than idealized grids, stochastic geometry is widely accepted as a tool to model cellular networks \cite{elsawy2013survey}. The past five years have witnessed a plethora of stochastic geometry based models that tackles different aspects in cellular networking \cite{elsawy2013survey, Sarabjot2013partition,jo2011outage,jo2012heterogeneous, elsawy2014uplink,dhillon2012uplink, mimo2013Di_Renzo, modeling_MIMO_ASE_harpreet, cao2012optimal, mukherjee2013energy, zzz, Hazem2015}. However,  the majority of these models do not account for mobility and focus on stationary users performance. For instance, coverage probability and rate performance are characterized for single-antenna downlink connections in \cite{Sarabjot2013partition,jo2011outage,jo2012heterogeneous}, for single-antenna uplink connections in \cite{elsawy2014uplink,dhillon2012uplink}, and for downlink connections with multiple antennas in \cite{mimo2013Di_Renzo, modeling_MIMO_ASE_harpreet}.  
Stochastic geometry also helps characterizing the performance of CP/UP split architecture in cellular networks. For instance, the energy efficiency gains provided by the CP/UP split architecture are characterized by Zakrzewska et al. \cite{mukherjee2013energy}. The effect of vertical offloading and BS sleeping on the energy efficiency for CU/UP split architecture is studied by
Zhang et al. \cite{zzz}. In our pervious work \cite{Hazem2015}, the throughput of the CU/UP split cellular architecture is studied.  However, none of the aforementioned studies incorporates the effect of mobility and handover into the analysis.

Surprisingly, few models can be found in the literature that exploit stochastic geometry to characterize mobility in cellular networks. The handover rate in cellular networks is first characterized by Lin et al. \cite{lin2013towards}, in which expressions for the handover rate are derived for random waypoint mobility model in a single-tier cellular network. The handover rate for multi-tier cellular networks is characterized by  Bao and Liang \cite{bao2015stochastic} for arbitrary mobility model. However, neither \cite{lin2013towards} nor \cite{bao2015stochastic} investigates the effects of handover on important performance metrics such as coverage, rate, or delay. The handover effects on coverage and rate  are  investigated by Sadr and Adve \cite{sadr2015handoff} for random way point mobility model. The authors derive the probability of handover and use the coverage probability for stationary users multiplied by a handover cost factor to infer the coverage probability for users experiencing handovers. Note that the handover cost factor in \cite{sadr2015handoff} is considered as a network parameter that reflects the SINR degradation during handovers.   Zhang et al. \cite{zhangdelay} investigate the effect of delay-reliability tradeoff in dense cellular networks for static and high mobility users under a time slotted transmission scheme. The authors show that high mobility users outperform  static users because mobile users experience uncorrelated SINRs across different time slots. However, the results in \cite{zhangdelay} may be misleading because the model only captures the positive impact of mobility and overlooks the performance degradation that may occur due to handover signaling and delay. Finally, Ge et al. \cite{ge2015user} develop a social-activity aware mobility model, denoted as the individual mobility model, to represent the users clustering behavior in a two-tier cellular network. Assuming a single social community, located at the origin, which is covered by densely deployed SBSs, the coverage probability inside and outside the social community as well as the probabilities to arrive, depart, and stay in the social community are derived. However, the analysis in \cite{ge2015user} is only valid for finite networks where the social community inhabits a non-negligible portion of the total network and overlooks the effect of handovers. It is worth mentioning that, similar to \cite{sadr2015handoff}, the authors of \cite{zhangdelay} and \cite{ge2015user} use the stationary SINR analysis to infer the coverage probability of moving users.

Different from the existing literature, our proposed mobility-aware paradigm captures the handover effect on the users throughput in conventional and CP/UP network architectures. Different from \cite{sadr2015handoff}, the handover cost is not assumed and is rigorously derived from the system model. Also, different from \cite{zhangdelay} and \cite{ge2015user}, the developed model accounts for the handover effect and is not tailored to a specific mobility model.  Furthermore, the developed model accounts for signaling overhead, flexible user association scheme via association biasing, the availability of X2 interface between BSs, and  almost blank subframes (ABS) coordination between MBSs and SBSs.

\section{System Model and Assumptions}
\label{System_Model_and_Assumption}
In this section, we describe the network and mobility models and assumptions.
\subsection{Network Model}
\normalsize
We consider a two-tier downlink cellular network with BSs in each tier modeled via an independent two dimensional homogeneous Poisson point process (PPP) $\mathbf{\Phi}_k$ of density $\lambda_k$, where $k \in \{1,2\}$. The macro cell tier and small cell (phantom cell) tier are denoted by $k=1$ and $k=2$,  respectively.  Mobile users are spatially distributed according to an independent PPP $\mathbf{\Phi}_u$ with density $\lambda^{(u)}$. All BSs in the $k^{th}$ tier are equipped with single antennas, transmit with the same power $P_k$, and always have packets to transmit.  We consider a general power law path loss model, with  path loss exponent $\alpha_k$, for both desired and interference downlink signal powers. Furthermore, signal attenuation due to multi-path fading is modeled using an independent Rayleigh distribution such that the channel power gain $H_{x}\sim \exp(1)$. {\color{black}A list of the key mathematical notations used in this paper is given in Table \ref{Notation_Summary}}.

Due to the transmission power disparity between the two tiers, the BSs footprints are represented by a weighted Poisson Voronoi diagram \cite{ash1986generalized} as depicted in Fig. \ref{weighted_Voronoi}. To enable flexible cell association and fine-grained control of  BS loads, we follow the model in \cite{jo2012heterogeneous} and introduce the bias factor $B$ to artificially encourages/discourages users to associate with the small cell tier.

\begin{table}
\centering
\caption{{\color{black} Mathematical Notations} }
\resizebox{0.45 \textwidth}{!}{\begin{tabular}{|c|c|}
\hline
  \textbf{Notation}&\textbf{Description}\\ \hline
  $\mathbf{\Phi}_k$;$\mathbf{\Phi}_u$& PPP of BSs of $k^{th}$ tier; PPP of mobile users. \\ \hline
   $\lambda_k$;$\lambda^{(u)}$& Density of BSs of $k^{th}$ tier; density of mobile users.\\ \hline
   $P_k$&\shortstack{Transmit power of BSs of $k^{th}$ tier.}\\ \hline
  B&\shortstack{Association bias for $2^{nd}$ tier.}\\ \hline
  $\mathbf{\alpha}_k$&\shortstack{Path loss exponent of $k^{th}$ tier.}\\ \hline
  $\mathcal{V}$&\shortstack{\footnotesize Mobile user velocity.}\\ \hline
  \\ [-1.5em]
  $HO_{ij}^{(c)}$&\shortstack{ Mean number of handovers per unit length from \\tier $i$ to $j$, for conventional network.}\\ \hline
    \\ [-1.5em]
  $MHO^{(s)}$&\shortstack{ Mean number of inter-anchor handovers per\\ unit length for CP/UP split network.}\\ \hline
    \\ [-1.5em]
  $VHO^{(s)}$&\shortstack{ Mean number of intra-anchor handovers per\\  unit length for CP/UP split network.}\\ \hline
    \\ [-1.5em]
  $\mathbf{\eta}$&\shortstack{ Fraction of time dedicated to serve biased mobile\\ users with no interference from the macro tier.}\\ \hline
    \\ [-1.5em]
  $\mathbf{\mu_C}$&\shortstack{ Control data overhead fraction in \\overall network capacity. } \\ \hline
  $\theta$&\shortstack{ Predefined threshold for correct signal reception.} \\ \hline
  $D^{(c)}_{HO}$ &\shortstack{ Handover cost in conventional network.} \\ \hline
   $D^{(s)}_{HO}$ &\shortstack{ Handover cost in CP/UP split network.} \\ \hline
     \\ [-1.5em]
   $\mathcal{X};\mathcal{Z}$ &\shortstack{ Probability of having X2 interface in conventional;\\ and CP/UP split architecture handovers.} \\ \hline
     \\ [-1.5em]
  $d^{(c)};\tilde{d}^{(c)}$&\shortstack{Delay per non X2 handover;delay\\per X2 handover in conventional network}. \\ \hline
    \\ [-1.5em]
  $d^{(s)}_{m};\tilde{d}_m^{(s)}$& \shortstack{ Inter-anchor handover delay without X2 interface;Inter-anchor \\ handover delay with X2 handover in CP/UP split network.}\\ \hline
  $d^{(s)}_{v}$& \shortstack{ Intra-anchor handover delay for CP/UP split network.} \\ \hline
    \\ [-1.5em]
  $\mathcal{\Large\emph{u}}_j$&$\shortstack{ Macro cell users \emph{j} = 1, small cell users\\ \emph{j} =2, biased small cell users \emph{j} = B}$.\\ \hline
  $\gamma$ &\shortstack{ Control signaling reduction factor}\\ \hline
  $AT^{(c)}$&\shortstack{ Average per-user throughput in the conventional network.} \\ \hline
  $AT^{(s)}$&\shortstack{ Average per-user throughput in CP/UP split network.} \\ \hline
  $\mathcal{A}_j$& Association probability of a typical user $\mathcal{\Large\emph{u}}_j$ .\\ \hline
    \\ [-1.5em]
  $\mathcal{T}^{(c)}_j$&\shortstack{ BS throughput in each association \\ state category for conventional network.} \\ \hline
    \\ [-1.5em]
  $\mathcal{T}^{(s)}_j$&\shortstack{ BS throughput in each association \\ state category for CP/UP split network.} \\ \hline
  $\mathcal{SE}^{(c)}$&\shortstack{ Spectral efficiency for conventional network.}\\ \hline
  $\mathcal{SE}^{(s)}$&\shortstack{ Spectral efficiency for CP/UP split network.}\\ \hline
  \footnotesize$P_{12}$;$P_{21}$&\shortstack{ $P_{12}=\frac{P_{1}}{P_{2}}$; $P_{21}=\frac{1}{P_{12}}$.}\\ \hline
  $\tilde{P}_{12}$;$\tilde{P}_{21}$&{ $\tilde{P}_{12}=\frac{P_{1}}{BP_{2}}$; $\tilde{P}_{21}=\frac{1}{\tilde{P}_{12}}$.}\\ \hline
  $\rho(a,b)$&\shortstack{ $\rho(a,b)=a+\sqrt{b}\arctan{(\sqrt{b})}$.}\\ \hline
  $\tilde{\lambda}_{k}$&\shortstack{ $\tilde{\lambda}_{k} = \frac{2 \pi \lambda_{k}}{\alpha_{k}-2}$.}\\ \hline
  $ \text{ }_{2}F_{1}(\cdot,\cdot;\cdot;\cdot)$&\shortstack{ The hypergeometric function.}\\ \hline
\end{tabular}}
\label{Notation_Summary}
\end{table}
\normalsize

Let $r_{k}$ denote the distance between an arbitrary mobile user and the nearest BS in the $k^{th}$ tier, then the biased association rule assigns a mobile user to the macro tier if $P_{1}r_{1}^{-\alpha_1} >  P_{2}B r_{2}^{-\alpha_2}$, and to a small (phantom) cell otherwise.  {\color{black}Based on the aforementioned association criterion and following the notation in \cite{Sarabjot2013partition},}  the complete set of users is divided into the following three non-overlapping sets:

\begin{equation}\label{sets}
  \hspace{-.5 cm} u \in \begin{cases}
    \mathcal{\Large\emph{u}}_1& \mbox{if } P_{1}r_{1}^{-\alpha_1} \geq P_{2}Br_{2}^{-\alpha_2}\\
    \mathcal{\Large\emph{u}}_{2}&\mbox{if } P_{2}r_{2}^{-\alpha_2} > P_{1}r_{1}^{-\alpha_1}\\
    \mathcal{\Large\emph{u}}_B&\mbox{if } P_{2}r_{2}^{-\alpha_2} \leq P_{1}r_{1}^{-\alpha_1} < P_{2}Br_{2}^{-\alpha_2}
  \end{cases}
\end{equation}
where $\mathcal{\Large\emph{u}}_1$ denotes the set of macrocell users, $\mathcal{\Large\emph{u}}_{2} $ denotes the set of non-biased small cell users, and $\mathcal{\Large\emph{u}}_B$ is the set of biased small cell users, where $\mathcal{\Large\emph{u}}_1\cup \mathcal{\Large\emph{u}}_{2} \cup \mathcal{\Large\emph{u}}_B= \mathbf{\Phi}_u$ and  $\mathcal{\Large\emph{u}}_1\cap \mathcal{\Large\emph{u}}_{2} \cap \mathcal{\Large\emph{u}}_B= \phi$.

\begin{figure}[!t]
\begin{center}
\scalebox{0.36}[0.36]{\includegraphics{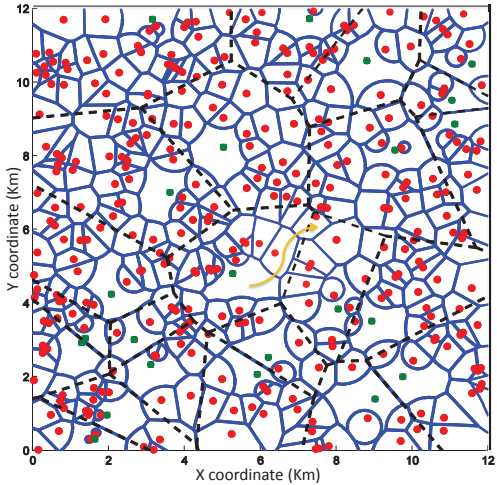}}
\end{center}
\caption{Two-tier weighted Possion Voronoi diagram representing a cellular network. The green squares and the red circles represent macro BSs and small BSs, respectively. The figure shows a user's trajectory (highlighted in orange), intra-anchor handover boundaries (in blue) and inter-anchor handover boundaries (in dotted black) for the CP/UP split architecture.}
\vspace{-.2cm}
\label{weighted_Voronoi}
\end{figure}

We consider two modes of operation, namely the conventional and CP/UP split, as shown in Fig.~\ref{split_and_no_split}. In the conventional network architecture, we assume that the control overhead consumes $\mu_c$ of the data rate and that each user gets the control and data from the same BS. We also assume universal frequency reuse scheme with almost blank sub-frames (ABS) interference management between macro cells and biased small cells \cite{Sarabjot2013partition}.\footnote{Universal frequency reuse is considered for the conventional network architecture because it always results in higher user throughput than dedicated spectrum access as shown in \cite{Hazem2015}.}  That is, a fraction $\eta$ of  time is dedicated to serving biased mobile users (i.e., $\mathcal{\Large\emph{u}}_B$) with no interference from the macro tier (i.e., {\color{black}MBSs do not send or send with very low power during the ABSs interval}).  In the CP/UP split network architecture, each small cell user (i.e., each user in  $\mathcal{\Large\emph{u}}_2$ and   $\mathcal{\Large\emph{u}}_{B}$) has double association in which the SBS transmit data only and the control signaling overhead is communicated via the MBS. Note that the control overhead for small cell users in the CP/UP split case consumes $\mu_c/\gamma$ of the data rate, where $\gamma \geq 1$ is offered control reduction factor~\cite{hoymann2013lean, split2012phantom}. It is worth noting that the decoupled, but simultaneous,  data and control association of the CP/UP split architecture necessitates a dedicated spectrum assignment for each tier. To conduct a fair comparison, we assume both the conventional and CP/UP split architectures have the same available spectrum of $W$, however, the CP/UP architecture split $W$ into  $W_1$ and $W_2 = W-W_1$ for the macro and small-cell tiers, respectively.

\subsection{User Mobility}
\label{Spectrum_Allocation_Control_Burden}

We assume that each user moves with an arbitrary trajectory and velocity, in which a handover occurs when a user crosses over a cell boundary. {However, we assume that the overall users mobility model preserves the spatial uniformity of users across the network.} We define a vertical handover as one made between two BSs in two different tiers, and a horizontal handover as one made between two BSs in the same tier.

Fig.~\ref{weighted_Voronoi} shows the handover boundaries for the conventional and CP/UP split network architectures.  In the CP/UP split network architecture, {the black dotted Voronoi tessellation represents control handover boundaries and the blue weighted Voronoi tessellation represents the data handover boundaries. In the conventional network architecture, the blue weighted Voronoi tessellation represents both the data and control handover boundaries.

In the conventional network architecture, users change their association (i.e., control and data) upon each handover. All handovers are managed through mobility management entity (MME) in the core network if direct X2 interface is not available between the serving and target BSs. Otherwise, the handover signalling is performed via the X2 interface without involving the core network, which highly reduces the handover delay. The handovers that occurs in the conventional network architecture can be categorized into the following cases: (1) vertical handover from a MBS to a SBS, (2) vertical handover from a SBS to a MBS, (3)  horizontal handover between two MBSs, and (4) horizontal handover  between two SBSs. In the conventional network architecture, the mean number of handovers, from tier $i$ to tier $j$, that occurs per unit length of a user trajectory is denoted by $HO_{ij}^{(c)}$, where $i,j \in\{1,2\}$ and the superscript $(c)$ denotes the conventional network architecture.


In the CP/UP split network architecture, the MBSs function as mobility anchors for data handovers within macro-to-macro Voronoi tessellation (black dotted tessellation in Fig.~\ref{weighted_Voronoi}). That is, the weighted Voronoi tessellation constructed w.r.t. all BSs in all tiers determines the data plane association and the Voronoi tessellation constructed w.r.t. MBSs only determines the control signaling and handover support association as shown in Fig.~\ref{weighted_Voronoi}. Consequently, only two types of handover occurs in the CP/UP split architecture, namely, (1) intra-anchor handover, and (2) inter-anchor handover. 
An inter-anchor handover occurs when a user crosses the boundary between two MBSs, and the handover is managed via the MME in the core network when there is no X2 interface between the engaged MBSs. In contrast, an intra-anchor handover is always transparent to the MME and is managed via the anchor BS, which reduces the handover delay because the MME is not notified.  In the CP/UP split network architecture, we denote the mean number of inter-anchor and intra-anchor handovers per unit length of the user trajectory as $MHO^{(s)}$ and $VHO^{(s)}$, respectively, where the superscript (s) denotes the CP/UP split network. It is worth noting that users change their control association without changing their data association when crossing over a macro-boundary within the coverage of a SBS. This type of handover is treated as an inter-anchor handover because the MME is informed. 

For tractability, we assume that users' trajectories are long enough to go through all three association states $j \in {1,2, B}$.  We also use the spatially averaged signal-to-interference-plus-noise-ratio (SINR) for stationary users provided by a given tier to infer the average SINR experienced by a mobile device during the journey through that tier. {\color{black} This assumption is validated later in Section \ref{validation_and_result}}. In other words, we compute  ${\rm SINR}_j$ provided by tier $j$ for a randomly selected stationary user and assume that mobile users will experience an average  ${\rm SINR}_j$ during their trajectories  in the $j^{th}$ tier. 

It is worth noting that the average stationary SINR assumption is needed for model tractability and was used in \cite{sadr2015handoff, ge2015user, zhangdelay}. This assumption only ignores the spatial correlations between the SINR values along each trajectory. However, averages over all trajectories and all users under all network realization are still captured by the analysis.

\section{Conventional and CP/UP Split Transmission Rate Models}
\label{pectrum_Allocation_Control_Burden}



Using Shannon's formula to define the ergodic rate, the average throughput delivered by a MBS and SBS for non-biased and biased users in the conventional architecture can be expressed as follows:

\begin{align}
\mathcal{T}^{(c)}_1&=(1-\mu_{C})(1-\eta) W  \mathbb{E}[\ln(1+{\rm SINR}^{(c)}_1)]
\label{throughput_conventional_1}, \\
\mathcal{T}^{(c)}_2&=(1-\mu_{C})(1-\eta) W  \mathbb{E}[\ln(1+{\rm SINR}^{(c)}_2)]
\label{throughput_conventional_2}, \\
\mathcal{T}^{(c)}_B&=(1-\mu_{C})\eta W  \mathbb{E}[\ln(1+{\rm SINR}^{(c)}_B)].
\label{throughput_conventional_B}
\end{align}

\noindent Note that MBSs are active only for $1-\eta$ fraction of the time due to the ABS interference management. On the other hand, small BSs are active all the time in which $1-\eta$ fraction of the time is dedicated for non-biased users $\mathcal{\Large\emph{u}}_2$ and $\eta$ fraction of the time is dedicated for biased users $\mathcal{\Large\emph{u}}_B$.

The decoupled data and control associations and dedicated spectrum access eliminate the inter-tier interference in the CP/UP split operation and changes the statistical nature of the ${\rm SINR}$ in the CP/UP split network architecture when compared to the conventional network architecture. In particular, in the CP/UP split network architecture as shown in Fig. \ref{split_and_no_split}, we have five different SINRs to consider, namely, ${\rm SINR}^{(s)}_1$ for macrocell users, ${\rm SINR}^{(s)}_{d2}$ for non-biased users' data, ${\rm SINR}^{(s)}_{c2}$ for non-biased users' control signaling, ${\rm SINR}^{(s)}_{dB}$ for biased users' data and ${\rm SINR}^{(s)}_{cB}$ for biased users' control signaling.  Let $W_1$  be the spectrum assigned to the macro tier and $W_2 = W - W_1$ be the spectrum assigned to the phantom cell tier. In this case, the average throughput of a small (phantom) cell users is given by:

\begin{align}
 \mathcal{T}^{(s)}_2&=  (1-\eta) W_2 \mathbb{E}[\ln(1+{\rm SINR}^{(s)}_{d2})],  \label{throughput_virtualized_2} \\
  \mathcal{T}^{(s)}_B&= \eta W_2 \mathbb{E}[\ln(1+{\rm SINR}^{(s)}_{dB})].
 \label{throughput_virtualized_B}
 \end{align}

\noindent Although there is no cross-tier interference in the CP/UP split network architecture due to spectrum splitting,  time sharing still exists in (\ref{throughput_virtualized_2}) and (\ref{throughput_virtualized_B}) because the phantom BSs dedicate a fraction $\eta$  of time to serve biased users. Note that the control overhead $\mu_c$ does not appear in the above throughput expressions because all control overhead is offloaded to the macro cells. The average throughput delivered to the macrocell users, after reserving the resources for phantom cell control signaling, is characterized via the following lemma:

\begin{lemma}\label{rate_lemma}
Consider a two-tier cellular network with the CP/UP split architecture,  PPP macro BSs with density $\lambda_1$, PPP phantom BSs with density $\lambda_2$, and a control reduction factor $\gamma$. Then the average throughput delivered to the macrocell users after resources for control signaling for phantom cell users have been reserved is expressed as:

\small
\begin{align}\label{macro_rate_lemma}
\mathcal{T}^{(s)}_1 &= (1-\mu_{C}) \mathcal{R}^{(s)}_1  \left( 1- \frac{\lambda_2 \mu_{C}}{\lambda_1\gamma} \left( \frac{\mathcal{T}^{(s)}_2}{\mathcal{R}_{c2}}+ \frac{ \mathcal{T}^{(s)}_B}{\mathcal{R}_{cB}}\right) \right),
\end{align}
\normalsize

\noindent where $\mathcal{R}^{(s)}_1 = W_1 \mathbb{E}[\ln(1+{\rm SINR}^{(s)}_1)]$ is the ergodic rate for macrocell users, $\mathcal{R}_{c2}= W_1 \mathbb{E}[\ln(1+{\rm SINR}^{(s)}_{c2})]$ is the average rate at which the control data is delivered to non-biased phantom cell users, $\mathcal{R}_{cB}= W_1 \mathbb{E}[\ln(1+{\rm SINR}^{(s)}_{cB})]$ is the average rate at which the control data is delivered to biased phantom cell users, and $\mathcal{T}^{(s)}_2$ and  $\mathcal{T}^{(s)}_B$ are the throughputs of non-biased and biased phantom cell users given in  \eqref{throughput_virtualized_2} and  \eqref{throughput_virtualized_B}, respectively.
\end{lemma}
\begin{proof}
 See Appendix~\ref{rate_lemma_proof}
\end{proof}

It is worth mentioning that \eqref{macro_rate_lemma} implicitly assumes that the control overhead is always a fraction $\mu_c$ of the available data rate and is only reduced by a factor of $\gamma$ for phantom cell users. Eq. \eqref{macro_rate_lemma} also assumes that the user population is sufficiently dense so that each phantom BS always has non-biased and biased small cell users to serve.

A CP/UP split network is said to be {\em feasible} if the MBSs have sufficient  bandwidth to serve macrocell users and to  provide control signaling to phantom cell users. From Lemma~\ref{rate_lemma}, the feasibility of the CP/UP split architecture is given in the following corollary

\begin{corollary} \label{col}
The CP/UP split architecture is feasible if and only if
\small
\begin{align}\label{condition_col}
  \frac{\mathcal{T}^{(s)}_2}{\mathcal{R}_{c2}}+ \frac{ \mathcal{T}^{(s)}_B}{\mathcal{R}_{cB}} & \leq \frac{\lambda_1\gamma}{ \lambda_2 \mu_{C}},
  \end{align}
  or equivalently
  \begin{align}
(1- \eta) \frac{\mathbb{E}[\ln(1+{\rm SINR}^{(s)}_{d2})]}{\mathbb{E}[\ln(1+{\rm SINR}^{(s)}_{c2})]}+ \eta \frac{\mathbb{E}[\ln(1+{\rm SINR}^{(s)}_{dB})]}{\mathbb{E}[\ln(1+{\rm SINR}^{(s)}_{cB})]} & \leq \frac{W_1 \lambda_1\gamma}{W_2 \lambda_2 \mu_{C}}.
\end{align}
\end{corollary}

According to Corollary~\ref{col}, the feasibility of the CP/UP split architecture is mainly limited by the average SINR experienced by the phantom cell users in MBSs.  Corollary~\ref{col} also suggests possible factors that can be manipulated to ensure the feasibility of the CP/UP split architecture are bandwidth assignment, relative BS densities, and/or control reduction factors.
\subsection{Per-user Mobility-aware Throughput Model}

\normalsize
The above expressions give the expected throughput for a typical user without capturing the main effects of network densifications. To have a realistic assessment to the densification gains, both throughput gains and the handover effects should be incorporated into the analysis. On one hand, network densification shrinks the BSs footprint, which reduces the number of users served by each BS and increases the share each user gets from his serving BS's throughput. On the other hand, network densification shrinks the BSs footprint, which increases the handover rate and overhead.
 During handover execution, the user releases the serving BS session and establishes a new session with the target BS. We assume that no data is delivered during handover execution and only handover-related signaling is communicated to the user. To incorporate the handover delay into the throughput expressions, we first compute the handover cost\footnote{The handover cost is a dimensionless unit which is computed as delay $\left(\frac{\text{\emph{sec}}}{{\text{\emph{handover}}}}\right) \times \text{velocity} \left(\frac{\text{\emph{meter}}}{\text{\emph{sec}}}\right) \times \text{handover rate} \left(\frac{\text{\emph{handovers}}}{\text{\emph{meter}}}\right)$.} for the conventional and CP/UP split architectures, which is the average duration consumed in handovers per unit time. Then we eliminate the handover duration from the throughput expressions (\ref{throughput_conventional_1})-(\ref{macro_rate_lemma}). For the conventional network architecture, the handover cost is expressed as:

\footnotesize
\begin{equation}\label{conv_cost}
D^{(c)}_{HO} = \left(\left(1-\mathcal{X}\right)d^{(c)}+\mathcal{X}\tilde{d}^{(c)}\right) \mathcal{V} \sum_i \sum_j HO_{ij}^{(c)},
\end{equation}
\normalsize

\noindent where $d^{(c)}$ and $\tilde{d}^{(c)}$ are the delays incurred by non X2 interface handover and X2 interface handover, respectively. $\mathcal{X}$ is the probability that an X2 interface is available between the serving and target BSs. In the CP/UP split network architecture, the delay incurred by an inter-anchor handover is different from the delay incurred by an intra-anchor handover, because the intra-anchor handover is always transparent to the core network.\footnote{It is expected that removing  the core network delay from a handover in the CP/UP split network (intra-anchor handover) reduces the handover delay by 50\% compared to a handover in the conventional network \cite{mahmoodiusing}.} On the other hand, all inter-anchor handovers are managed through the MME in the core network unless an X2 interface is available. Therefore, the handover cost for the CP/UP split architecture is given by:

\footnotesize
\begin{equation}\label{cpup_cost}
D^{(s)}_{HO} =  \mathcal{V} \left( MHO^{(s)} \left(\left(1-\mathcal{Z}\right)d^{(s)}_m+\mathcal{Z}\tilde{d}_m^{(s)}\right) + VHO^{(s)} d^{(s)}_{v}\right),
\end{equation}
\normalsize

\noindent where $d^{(s)}_{m}$, $\tilde{d}_m^{(s)}$, and $d^{(s)}_{v}$ are the delays incurred by an inter-anchor handover without X2 interface, an inter-anchor handover with X2 interface, and an intra-anchor handover, respectively. $\mathcal{Z}$ is the probability that a direct X2 connectivity is available between the serving and target MBSs.

Incorporating the handover delay in to the throughput analysis, assuming that each BS uniformly distributes the resources across the users it serves, and using the law of total probability, the average per-user throughput along his trajectory for the conventional and CP/UP split architectures are, respectively, expressed as:

\footnotesize
\begin{equation}\label{watpu_c}
AT^{(c)}=\left(\frac{\mathcal{A}_{1}\mathcal{T}^{(c)}_1}{\mathcal{N}_1}+\frac{\mathcal{A}_{2}\mathcal{T}^{(c)}_2}{\mathcal{N}_2}+\frac{\mathcal{A}_{B}\mathcal{T}^{(c)}_B}{\mathcal{N}_B}\right)\left(1-\min\left(1,D^{(c)}_{HO}\right)\right), 
\end{equation}
\normalsize
\footnotesize
 \begin{equation}\label{watpu_v}
AT^{(s)}=\left(\frac{\mathcal{A}_{1}\mathcal{T}^{(s)}_1}{\mathcal{N}_1}+\frac{\mathcal{A}_{2}\mathcal{T}^{(s)}_2}{\mathcal{N}_2}+\frac{\mathcal{A}_{B}\mathcal{T}^{(s)}_B}{\mathcal{N}_B} \right)\left(1-\min\left(1,D^{(s)}_{HO}\right)\right),
\end{equation}
\normalsize

\noindent where $\mathcal{A}_j$ is the probability of being served by a BSs in $j \in \{1,2,B\}$ case and $\mathcal{N}_j$ is the expected number of users sharing the BS resources with the typical user in the $j \in \{1,2,B\}$ case. Note that $\mathcal{A}_j$  and $\mathcal{N}_j$ in  \eqref{watpu_c} and \eqref{watpu_v} are independent from the network architecture and  are calculated according to the association rule  \eqref{sets}. The effect of control signaling offloaded to the MBSs in the CP/UP split is already captured by $\mathcal{T}_j^{(s)}$.

 Eqs. \eqref{watpu_c} and \eqref{watpu_v} are the main performance metrics in this paper, which are the mobility aware per-user average throughput in the conventional and CP/UP split architectures. It is worth re-emphasizing that  \eqref{watpu_c} and \eqref{watpu_v} assume that the users have long trajectories, that each user passes through all association states during their trajectories, and that the mobility model preserves the users spatial uniformity across the network. It is worth mentioning that when the average cell dwell time becomes less than the handover delay, the handover costs in \eqref{conv_cost} and \eqref{cpup_cost} are greater than unity. Consequently, the network fails to support users and the average throughputs in \eqref{watpu_c} and \eqref{watpu_v} are nullified.

 Exploiting the long trajectories and users spatial uniformity, the association probabilities and BS loads can be obtained by following \cite{Sarabjot2013partition} and \cite{jo2012heterogeneous}, respectively. In particular, the association probabilities $\mathcal{A}_{1}$,  $\mathcal{A}_{2}$ and $\mathcal{A}_{B}$ can be viewed as the percentages of the $\mathbb{R}^2$ domain served by the MBSs, the unbiased SBS, and the biased SBS, respectively. Consequently, the association probabilities are given by \cite{Sarabjot2013partition}:

\footnotesize
\begin{equation}\label{A1}
\mathcal{A}_{1}=2\pi\lambda_{1}\int_{0}^\infty r\mbox{ exp }\left(-\pi\left(\lambda_{1}r^{2}+\lambda_{2}\tilde{P}_{21} ^{\frac{2}{\alpha_{2}}}r^{\frac{2\alpha_{1}}{\alpha_{2}}} \right)\right)\mbox{d}r,
\end{equation}
\begin{equation}\label{A2}
\mathcal{A}_{2}=2\pi\lambda_{2}\int_{0}^\infty r\mbox{ exp }\left(-\pi\left(\lambda_{2}r^{2}+\lambda_{1}P_{12}^{\frac{2}{\alpha_{1}}}r^{\frac{2\alpha_{2}}{\alpha_{1}}} \right)\right)\mbox{d}r,
\end{equation}
\begin{multline}\label{A3}
\mathcal{A}_{B}=2\pi\lambda_{2}\int_{0}^\infty r\Bigg\{ \mbox{ exp }\left[-\pi\left( \lambda_{1}\left(\tilde{P}_{12}^{\frac{2}{\alpha_{1}}}r^{\frac{2\alpha_{2}}{\alpha_{1}}}+\lambda_{2}r^{2}   \right)\right)\right] \\ - \mbox{ exp } \left[-\pi\left( \lambda_{1}\left(P_{12}^{\frac{2}{\alpha_{1}}}r^{\frac{2\alpha_{2}}{\alpha_{1}}}+\lambda_{2}r^{2}   \right)\right)\right] \Bigg\}\mbox{d}r,
\end{multline}
\normalsize
\noindent {\color{black} where $P_{12}=\frac{P_{1}}{P_{2}}$, $P_{21}=\frac{1}{P_{12}}$, $\tilde{P}_{12}=\frac{P_{1}}{BP_{2}}$, and $\tilde{P}_{21}=\frac{1}{\tilde{P}_{12}}$}.

Since the user spatial uniformity is preserved, the average number of users sharing the resources with the typical user for each of the association cases is computed as follows \cite{Sarabjot2013partition}:

\footnotesize
\begin{equation}\label{load}
   \mathcal{N}_{j}=\frac{1.28\lambda^{(u)}\mathcal{A}_{j}}{\lambda_{J(j)}}+1\notag,
 \end{equation}
 \normalsize

\noindent where $J(j)$ is a map from user set association index $j \in \{1,2,B\}$ to serving tier index $k \in \{1,2\}$ as follows: $J(1) = 1$, $J(2) = J(B) = 2$. Therefore,

\footnotesize
\begin{align}
\mathcal{N}_{1}&=1.28\left(2\pi\lambda^{(u)}\int_{0}^{\infty}r\mbox{ exp}\Bigg\{-\pi\left[\lambda_{1}r^{2}+\lambda_{2}\tilde{P}_{21}^{\frac{2}{\alpha_{2}}}r^{\frac{2\alpha_{1}}{\alpha_{2}}}\right]\Bigg\}\mbox{d}r\right)+1,\notag \\
\mathcal{N}_{2}&=1.28\left(2\pi\lambda^{(u)}\int_{0}^{\infty}r\mbox{ exp}\Bigg\{-\pi\left[\lambda_{1}P_{12}^{\frac{2}{\alpha_{1}}}r^{\frac{2\alpha_{2}}{\alpha_{1}}}+\lambda_{2}r^{2}\right]\Bigg\}\mbox{d}r\right)+1,\notag\\
\mathcal{N}_{B}&=1.28\biggl(2\pi\lambda^{(u)}\int_{0}^{\infty}r\Bigg[\mbox{ exp}\left(-\pi\left(\lambda_{1}P_{12}^{\frac{2}{\alpha_{1}}}r^{\frac{2\alpha_{2}}{\alpha_{1}}}+\lambda_{2}r^{2}\right) \right) \notag \\&-\mbox{ exp}\left(-\pi \left(\lambda_{1}\tilde{P}_{12}^{\frac{2}{\alpha_{1}}}r^{\frac{2\alpha_{2}}{\alpha_{1}}}+\lambda_{2}r^{2}  \right)\right)\Bigg]\mbox{d}r\biggr)+1\notag.
\end{align}
\normalsize

An important scenario of interest is the case of equal path-loss exponents ($\alpha_1=\alpha_2= 4$), which not only simplifies the analysis but also a practical value for outdoor cellular communications in urban environments \cite{elsawy2013survey, Sarabjot2013partition,jo2011outage,jo2012heterogeneous, elsawy2014uplink,dhillon2012uplink, mimo2013Di_Renzo, modeling_MIMO_ASE_harpreet, cao2012optimal, mukherjee2013energy, zzz, Hazem2015}. In this case, the association probabilities and BS loads reduce to:

\small
\begin{align}
\mathcal{A}_{1}&=\frac{\lambda_{1}}{\lambda_{1}+\lambda_{2} \sqrt{\tilde{P}_{21}}},\notag
\mathcal{A}_{2}=\frac{\lambda_{2}}{\lambda_{1}\sqrt{P_{12}}+\lambda_{2}},\notag\\ \mathcal{A}_{B}&=\frac{\lambda_{2}}{\lambda_{1}\sqrt{\tilde{P}_{12}}+\lambda_{2}}-\frac{\lambda_{2}}{\lambda_{1}\sqrt{P_{12}}+\lambda_{2}},
\label{assoc_spec}
\end{align}
\normalsize
and
\vspace{2mm}
\small
{$ \mathcal{N}_{1} =\frac{1.28\lambda^{(u)}}{\lambda_{1}+\lambda_{2}\sqrt{\tilde{P}_{21}}}+1$}, {$\mathcal{N}_{2} = \frac{1.28\lambda^{(u)}}{\lambda_{1}\sqrt{P_{12}}+\lambda_{2}}+1$},

{$\mathcal{N}_{B} = 1.28\left(\frac{\lambda^{(u)}}{\lambda_{1}\sqrt{\tilde{P}_{12}}+\lambda_{2}}-\frac{\lambda^{(u)}}{\lambda_{1}\sqrt{P_{12}}+\lambda_{2}}\right)+1$}.
\normalsize
\vspace{2mm}

The missing components to calculate \eqref{watpu_c} and \eqref{watpu_v} are the spectral efficiencies (i.e., $\mathcal{SE} = \mathbb{E}[\ln(1+{\rm SINR})]$) and handover cost (i.e., $D_{HO}$), which are characterized in Section~\ref{Performance_Analysis} and Section~\ref{mobility_analysis}, respectively.



\section {SINR and Spectral Efficiency Characterization}
\label{Performance_Analysis}

As mentioned earlier, for tractability, we use the spatially averaged spectral efficiency for stationary users to infer the average spectral efficiency for mobile users. This assumption is validated later in Section \ref{validation_and_result} and shown to give accurate approximation for the SINR distribution.

To characterize the SINR, and hence the spectral efficiency,  we first characterize the service distance distribution. Then, we characterize the {\rm SINR} and spectral efficiency for both the conventional and CP/UP split network architectures.



\subsection{Service Distances}
\label{distance_and_association}


As shown in Fig.~\ref{split_and_no_split}, we need to characterize five service distances, namely ${R_{1}}$, ${R_{2}}$, ${R_{B}}$, ${R_{c2}}$, and ${R_{cB}}$. The conventional service distances ${R_{1}}$, ${R_{2}}$, ${R_{B}}$, which are for  users in $\Large\emph{u}_1$, $\Large\emph{u}_2$, and $\Large\emph{u}_B$ respectively, are characterized in \cite{Sarabjot2013partition}. Conditioned on the association, the probability density functions (PDFs) of the  ${R_{1}}$, ${R_{2}}$, and ${R_{B}}$ are given by

\footnotesize
\begin{align} \label{pdf_distance}
\mathbf{f}_{R_{1}}(r)=&\frac{2\pi\lambda_{1}}{\mathcal{A}_{1}}r e^{-\pi\left(\lambda_{1}r^{2}+\lambda_{2}\tilde{P}_{21}^{\frac{2}{\alpha_{2}}}r^{\frac{2\alpha_{1}}{\alpha_{2}}} \right)}, r \geq 0,
\end{align}
\begin{align} \label{r2}
&\mathbf{f}_{R_{2}}(r)=\frac{2\pi\lambda_{2}}{\mathcal{A}_{2}}r e^{-\pi\left(\lambda_{2}r^{2}+\lambda_{1}P_{12}^{\frac{2}{\alpha_{1}}}r^{\frac{2\alpha_{2}}{\alpha_{1}}} \right)}, r \geq 0,
\end{align}
\begin{multline} \label{rb}
\mathbf{f}_{R_{B}}(r)=\frac{-2\pi\lambda_{2}}{\mathcal{A}_{B}}r\Bigg[ e^{-\pi\left(\lambda_{1}P_{12}^{\frac{2}{\alpha_{1}}}r^{\frac{2\alpha_{2}}{\alpha_{1}}}+\lambda_{2}r^{2}\right)}-\\ e^{-\pi \left(\lambda_{1}\tilde{P}_{12}^{\frac{2}{\alpha_{1}}}r^{\frac{2\alpha_{2}}{\alpha_{1}}}+\lambda_{2}r^{2}  \right)}\Bigg], r \geq 0.
\end{multline}
\normalsize

As shown in Fig.~\ref{split_and_no_split}, the association for $\Large\emph{u}_1$, and the data association for $\Large\emph{u}_2$ and $\Large\emph{u}_B$ in the CP/UP split case have similar distribution to  ${R_{1}}$, ${R_{2}}$, and ${R_{B}}$ given in \eqref{pdf_distance}, \eqref{r2}, and \eqref{rb}, respectively. The distributions for control link distances of the CP/UP split architecture are  given by the following lemma

\begin{lemma} \label{lem_distances}
Let $R_{c2}$ and $R_{cB}$ denote the distances from the MBS that provides the control signaling to $\Large\emph{u}_2$ and $\Large\emph{u}_B$, respectively, in a cellular network with the CP/UP split architecture. Then the distributions of $R_{c2}$ and $R_{cB}$  are given by

\footnotesize
\begin{align}
&\mathbf{f}_{R_{c2}}(r) = \frac{2 \pi \lambda_1 r}{\mathcal{A}_2}  \left(  e^{-\pi \lambda_1 r^2}  -  e^{-\pi \left(\lambda_1 r^2+ \lambda_2 P_{21}^{\frac{2}{\alpha_1}} r^{\frac{2 \alpha_2}{\alpha_1}} \right)} \right), r \geq 0,
\end{align}
\normalsize

\footnotesize
\begin{align}
& \!\!\!\!\!\!\!\!\mathbf{f}_{R_{cB}}(r)=\frac{2 \pi \lambda_1 r}{\mathcal{A}_B}  \left(  e^{-\pi \left(\lambda_1 r^2+ \lambda_2 P_{21}^{\frac{2}{\alpha_1}} r^{\frac{2 \alpha_2}{\alpha_1}} \right)}\right. \notag \\
& \left. \quad  \quad  \quad  \quad  \quad  \quad  \quad  \quad  \quad  \quad - e^{-\pi \left(\lambda_1 r^2+ \lambda_2 \tilde{P}_{21}^{\frac{2}{\alpha_1}} r^{\frac{2 \alpha_2}{\alpha_1}} \right)} \right), r \geq 0.
\end{align}
\normalsize
\end{lemma}

\begin{proof}
See Appendix \ref{distances}.
\end{proof}

\subsection{Coverage Probability Analysis}
\label{ATR_label}

The coverage probability is defined by the complementary cumulative distribution function (CCDF) of the SINR (i.e, $\mathbb{P}[{\rm SINR}>\theta]$, where $\theta$ denotes the predefined threshold for correct signal reception). Without loss of generality, the {\rm SINR} analysis is performed for a {\color{black} test mobile user} located at the origin.   According to Slivnyak's theorem, all other users have statistical SINR properties equivalent to that of the test user located at the origin \cite{haenggi2009interference}. Therefore, the analysis holds for an arbitrary {\color{black}mobile user} located at any other location.

For the sake of exposition, we define four types of interferences caused by the BSs in $\Phi_1$ and $\Phi_2$ with respect to the origin, which are
\begin{itemize}
\item The interference from all MBSs $\mathcal{I}_1 = \sum\limits_{x\in\mathbf{\Phi_{1}}} P_{1}H_{x}x^{-\alpha_{1}}$.
\item The interference from all MBSs excluding the one nearest to the origin $\mathcal{I}^{o}_1 = \sum\limits_{x\in\mathbf{\Phi_{1}} \setminus x_o} P_{1}H_{x}x^{-\alpha_{1}}$.
\item The interference from all SBSs $\mathcal{I}_2 = \sum\limits_{x\in\mathbf{\Phi_{2}}} P_{2}H_{x}x^{-\alpha_{2}}$.
\item The interference from all SBSs excluding the one nearest to the origin $\mathcal{I}^{o}_2 = \sum\limits_{x\in\mathbf{\Phi_{2}} \setminus x_o} P_{2}H_{x}x^{-\alpha_{2}}$.
\end{itemize}



The SINR at the test user's location, the origin, can be defined as

\footnotesize
\begin{equation} \label{gen_sinr}
{\rm SINR} = \frac {P_{BS} H r_0^{-\alpha}  }{\mathcal{I}_{agg}+\sigma^2},
\end{equation}
\normalsize

\noindent where $P_{BS}$ is the serving BS transmit power, $H$ is the random channel power gain, $r_o$ is the distance between the test user and the serving BS, $\mathcal{I}_{agg}$ is the aggregate interference, and $\sigma^2$ is the noise power. The parameters in \eqref{gen_sinr} to compute the SINR experienced by the users in $\mathcal{\Large\emph{u}}_k$, $k \in \{1,2,B\}$ for the conventional and CP/UP split architectures are given in Table~\ref{SINR_Table}.

\footnotesize
\begin{table}[h]
\centering
\small
\caption{{\color{black} SINR Parameters} }
\resizebox{0.45 \textwidth}{!}{\begin{tabular}{|c|c|c|c|c|c|c|c|c|}
\hline
 \multirow{2}{*}{k}& \multicolumn{4}{|c|}{Conventional}  & \multicolumn{4}{|c|}{CP/UP split} \\ \cline{2-9}
& \textbf{SINR} & $P_{BS}$ & $r_o$ &$\mathcal{I}_{agg}$ & \textbf{SINR} &$P_{BS}$ & $r_o$ &$\mathcal{I}_{agg}$ \\ \hline \hline
\multirow{2}{*}{1}& \multirow{2}{*}{ ${\rm SINR}^{(c)}_{1}$} &  \multirow{2}{*}{$P_1$}  & \multirow{2}{*}{$R_1$} & \multirow{2}{*}{$\mathcal{I}^o_1 + \mathcal{I}_2$} &   \multirow{2}{*}{${\rm SINR}^{(s)}_{1}$} & \multirow{2}{*}{$P_1$}  & \multirow{2}{*}{$R_1$} & \multirow{2}{*}{$\mathcal{I}^o_1$} \\
 & &   &  &  &  &  &  & \\ \hline
\multirow{2}{*}{2}&  \multirow{2}{*}{ ${\rm SINR}^{(c)}_{2}$} &  \multirow{2}{*}{$P_2$}  & \multirow{2}{*}{$R_2$} & \multirow{2}{*}{$\mathcal{I}_1 + \mathcal{I}^o_2$} &  ${\rm SINR}^{(s)}_{d2}$ &$P_2$  & $R_2$ & $ \mathcal{I}^{o}_2$ \\   \cline{6-9}
 & &   &  &  &  ${\rm SINR}^{(s)}_{c2}$ &$P_1$  & $R_{c2}$ & $\mathcal{I}^o_1 $ \\ \hline
\multirow{2}{*}{B}&  \multirow{2}{*}{ ${\rm SINR}^{(c)}_{B}$} &  \multirow{2}{*}{$P_2$}  & \multirow{2}{*}{$R_B$} & \multirow{2}{*}{$ \mathcal{I}^{o}_2$} &  ${\rm SINR}^{(s)}_{dB}$ &$P_2$  & $R_{B}$ & $\mathcal{I}^{o}_2$ \\   \cline{6-9}
 & &   &  &  &  ${\rm SINR}^{(s)}_{cB}$ &$P_1$  & $R_{cB}$ & $\mathcal{I}^o_1 $ \\ \hline
\end{tabular}}
\label{SINR_Table}
\end{table}
\normalsize

As shown in Table~\ref{SINR_Table}, in the conventional network architecture, the macro and non-biased small cells users experience inter-tier interference, which is due to the employed universal frequency reuse scheme. In contrast, the dedicated spectrum accesses employed by the CP/UP split architecture eliminates the inter-tier interference. Note that the biased users do not experience inter-tier interference in the conventional network architecture due to the ABS interference coordination employed by the MBSs. Table~\ref{SINR_Table} also shows the different SINR experienced by the data and control links in the CP/UP split architecture, which is due to the employed decoupled data and control associations.

For a predefined threshold reception $\theta$, the coverage probability $\mathcal{C}=\mathbb{P}[{\rm SINR}>\theta]$ of all users are characterized by the following lemma.

\begin{lemma}The {\rm SINR} coverage for the conventional network is given by
\label{average_trans_rate_three_types}

\footnotesize
{\color{black}
\begin{align}\label{macro_cell_coverage_general_11}
&\mathcal{C}^{(c)}_{1}=\int_{0}^\infty{\exp} \Bigg(-\tilde{\lambda}_{1}r^{2}\theta\text{ }_{2}F_{1}\left(1,1-\frac{2}{\alpha_{1}};2-\frac{2}{\alpha_{1}};-\theta\right) \notag \\&-\frac{\tilde{\lambda}_{2}}{B}r^{2}\theta \tilde{P}_{21}^{\frac{2}{\alpha_{2}}} \text{ }_{2}F_{1}\left(1,1-\frac{2}{\alpha_{2}};2-\frac{2}{\alpha_{2}};\frac{-\theta}{B}\right)\Bigg) \mathbf{f}_{R_{1}}(r) \text{d}r,
\end{align}}
{\color{black}
\begin{align}\label{small_cell_rate_bias_general_11}
\mathcal{C}^{(c)}_{2}&=\int_{0}^\infty \exp\left(-\tilde{\lambda}_{1}r^{2}\theta P_{12}^{\frac{2}{\alpha_{1}}}\text{ }_{2}F_{1}\left(1,1-\frac{2}{\alpha_{1}};2-\frac{2}{\alpha_{1}};-\theta\right)\right. \notag \\
&\left. -\tilde{\lambda}_{2}r^{2}\theta\text{ }_{2}F_{1}\left(1,1-\frac{2}{\alpha_{2}};2-\frac{2}{\alpha_{2}};-\theta\right)
\right)  \mathbf{f}_{R_{2}}(r) \text{d}r,
\end{align}
}
\footnotesize
{\color{black}
\begin{multline}\label{small_cell_rate_without_bias_general}
\mathcal{C}^{(c)}_{B} = \\\int_{0}^\infty\mbox{exp}\left(-\tilde{\lambda}_{2}r^{2}\theta\text{ }_{2}F_{1}\left(1,1-\frac{2}{\alpha_{2}};2-\frac{2}{\alpha_{2}};-\theta\right)\right)  \mathbf{f}_{R_{B}}(r) \mbox{d}r.
\end{multline}
}
\normalsize

The {\rm SINR} coverage for macrocell users in the CP/UP split network architecture is given by

\footnotesize
{\color{black}
\begin{align}\label{macro_cell_coverage_general}
\mathcal{C}^{(s)}_{1}&=\int_{0}^\infty{\exp}\left(-\tilde{\lambda}_{1}r^{2}\theta\text{ }_{2}F_{1}\left(1,1-\frac{2}{\alpha_{1}};2-\frac{2}{\alpha_{1}};-\theta\right) \right) \mathbf{f}_{R_{1}}(r) \text{d}r.
\end{align}
}
\normalsize
The {\rm SINR} coverage for the data connections for non-biased phantom cell users is given by
\footnotesize
{\color{black}
\begin{align}\label{small_cell_rate_bias_general}
\mathcal{C}^{(s)}_{d2}&=\int_{0}^\infty \exp\left(-\tilde{\lambda}_{2}r^{2}\theta\text{ }_{2}F_{1}\left(1,1-\frac{2}{\alpha_{2}};2-\frac{2}{\alpha_{2}};-\theta\right)\right)  \mathbf{f}_{R_{2}}(r) \text{d}r, \end{align}
}
\normalsize

\noindent and the biased small cell users $\mathcal{C}^{(s)}_{dB} = \mathcal{C}^{(c)}_{B}$ given in \eqref{small_cell_rate_without_bias_general}. The SINR coverage probabilities for the control links of the non-biased and biased phantom cell users  are given by

\footnotesize
\begin{align}\label{small_cell_rate_bias_general}
&\mathcal{C}^{(s)}_{c2}  =   \int_{0}^\infty \exp\left(-\tilde{\lambda}_{1}r^{2}\theta\text{ }_{2}F_{1}\left(1,1-\frac{2}{\alpha_{1}};2-\frac{2}{\alpha_{1}};-\theta\right)\right)  \mathbf{f}_{R_{c2}}(r) \text{d}r, \end{align}


\begin{multline}\label{small_cell_rate_without_bias_general_123}
\mathcal{C}^{(c)}_{cB}=\\\int_{0}^\infty\mbox{exp}\left(-\tilde{\lambda}_{1}r^{2}\theta\text{ }_{2}F_{1}\left(1,1-\frac{2}{\alpha_{1}};2-\frac{2}{\alpha_{1}};-\theta\right)\right)  \mathbf{f}_{R_{cB}}(r) \mbox{d}r,
\end{multline}

\normalsize
\normalsize

\noindent where {\color{black}$\tilde{\lambda}_{k} = \frac{2 \pi \lambda_{k}}{\alpha_{k}-2}$},$ \text{ }_{2}F_{1}(\cdot,\cdot;\cdot;\cdot)$ is the hypergeometric function, and $\mathbf{f}_{R_{1}}(r)$, $\mathbf{f}_{R_{2}}(r)$, $\mathbf{f}_{R_{B}}(r)$, $\mathbf{f}_{R_{c2}}(r)$, and $\mathbf{f}_{R_{cB}}(r)$  are given in Section \ref{distance_and_association}.

\end{lemma}
\begin{proof}
See Appendix \ref{Average transmission rate_prof}.
\end{proof}

For the special case of equal path-loss exponents $\alpha_1=\alpha_2=4$, the coverage probabilities reduce to the  simple closed-form expressions shown below:

\footnotesize
\begin{align}\label{simulation_snir_1}
\mathcal{C}^{(c)}_{1}=\frac{\lambda_{1}+\lambda_{2}\sqrt{\tilde{P}_{21}}}{\lambda_{1}\rho(1,\theta)+ \lambda_{2}\sqrt{\tilde{P}_{21}}\rho(1,\frac{\theta}{B})},
\end{align}

\begin{align}
\mathcal{C}^{(c)}_{2}=\frac{1}{\rho(1,\theta)},
\end{align}

\begin{align}
\mathcal{C}^{(s)}_{1}=\frac{\lambda_{1}+\lambda_{2}\sqrt{\tilde{P}_{21}}}{\lambda_{1}\rho(1,\theta)+\lambda_{2}\sqrt{\tilde{P}_{21}}},
\end{align}
\begin{align}
\mathcal{C}^{(s)}_{d2}=\frac{ \lambda_{2}+\lambda_{1}\sqrt{P_{12}}}{\lambda_{2}\rho(1,\theta)+\lambda_{1}\sqrt{P_{12}}},
\end{align}

\begin{align}
\mathcal{C}^{(s)}_{c2}=\frac{\bigg(1+\frac{\lambda_{1}}{\lambda_{2}}\sqrt{P_{12}} \bigg) \bigg( 1 - \frac{1}{1+\frac{\lambda_{2}}{\lambda_1} \sqrt{P_{21}} \rho(1,\theta)^{-1} }\bigg)}{\rho(1,\theta)},
\end{align}
\normalsize

\footnotesize
\begin{align}
\mathcal{C}^{(s)}_{dB} = \mathcal{C}^{(c)}_{B}&=\frac{\lambda_{2}}{\mathcal{A}_{B}} \left( \frac{\lambda_{1}\left(\sqrt{P_{12}}-\sqrt{ \tilde{P}_{12}}\right)}{\left(\lambda_{2} \rho(1,\theta)+\lambda_{1}\sqrt{\tilde{P}_{12}}\right)\left(\lambda_{2} \rho(1,\theta)+\lambda_{1}\sqrt{P_{12}}\right)}\right), \end{align}

\begin{align}\label{simulation_snir_2}
\mathcal{C}^{(s)}_{cB} &=\frac{\lambda_{1}}{\mathcal{A}_{B}} \left( \frac{\lambda_{2}\left(\sqrt{\tilde{P}_{21}}-\sqrt{P_{21}}\right)}{\left(\lambda_{1} \rho(1,\theta)+\lambda_{2}\sqrt{P_{21}}\right)\left(\lambda_{1}\rho(1,\theta)+\lambda_{2}\sqrt{\tilde{P}_{21}}\right)}\right), \end{align}
\normalsize
\noindent{\color{black} where $\rho(a,b)=a+\sqrt{b}\arctan\left(\sqrt{b}\right)$}.
\subsection{Spectral Efficiency Analysis}

The spectral efficiency is one of the main parameters to calculate the throughput of the conventional and CP/UP split users throughputs as shown in Section~\ref{pectrum_Allocation_Control_Burden}. The spectral efficiency ($\mathcal{SE} =  \mathbbm{E}[\ln(1+{\rm SINR})]$) can be directly derived from the coverage probability as follows

\footnotesize
\begin{align}
\mathcal{SE} = \mathbbm{E}[\ln(1+{\rm SINR})]&\stackrel{(a)}{=}\int_{0}^\infty\mathbb{P}[{\ln(1+ \rm{SINR})>\zeta}]\mbox{d}\zeta \notag \\
&=\int_{0}^\infty\mathbb{P}[{\rm SINR}>(e^\zeta-1)]\mbox{d}\zeta \notag \\
&\stackrel{(b)}{=}\int_{0}^\infty\frac{\mathbb{P}[{\rm SINR}>t]}{t+1}\mbox{d}t,
\label{xdr}
\end{align}
\normalsize

\noindent where (a) follows because $\ln(1+{\rm SINR})$ is a strictly positive random variable, and (b) follows by substituting variable $t= e^\zeta-1$. For general path loss exponent, the spectral efficiencies for macro-cell and small-cell users in the shared spectrum access scheme in the conventional network are given by \eqref{new_spec1} and \eqref{new_spec2}, respectively.  For the dedicated spectrum access scheme in the CP/UP split RAN, the spectral efficiencies are given by:

\begin{figure*}
\footnotesize
\begin{align} \label{new_spec1}
\!\!\!\!\!\!\!\!\!\!\!\!\!\!\!\!\!\!\mathcal{SE}^{(c)}_{1}=\int\limits_{0}^\infty \int\limits_{0}^\infty \frac{{\exp}\Bigg(-\tilde{\lambda}_{1}r^{2}t\text{ }_{2}F_{1}\left(1,1-\frac{2}{\alpha_{1}};2-\frac{2}{\alpha_{1}};-t\right) -\tilde{\lambda}_{2} r^{2}tP_{21}^{\frac{2}{\alpha_{2}}}B^{\frac{2}{\alpha_{2}}-1}\text{ }_{2}F_{1}\left(1,1-\frac{2}{\alpha_{2}};2-\frac{2}{\alpha_{2}};\frac{-t}{B}\right)\Bigg)}{t+1} \mathbf{f}_{R_{1}}(r) \text{d}r \text{d}t,
\end{align}
\hrule
\begin{align} \label{new_spec2}
\!\!\!\!\!\!\!\!\!\!\!\!\!\!\!\!\!\!\mathcal{SE}^{(c)}_{2}&=\int\limits_{0}^\infty \int\limits_{0}^\infty \frac{\exp\left(-\tilde{\lambda}_{1}r^{2}tP_{12}^{\frac{2}{\alpha_{1}}}\text{ }_{2}F_{1}\left(1,1-\frac{2}{\alpha_{1}};2-\frac{2}{\alpha_{1}};-t\right) -\tilde{\lambda}_{2}r^{2}t\text{ }_{2}F_{1}\left(1,1-\frac{2}{\alpha_{2}};2-\frac{2}{\alpha_{2}};-t\right)\right)}{t+1}  \mathbf{f}_{R_{2}}(r) \text{d}r \text{d}t.
\end{align}
\normalsize
\hrule
\end{figure*}

\footnotesize
{\color{black}
\begin{align} \label{new_spec3}
&\mathcal{SE}^{(s)}_{1} = \notag \\
&\int\limits_{0}^\infty \int\limits_{0}^\infty \frac{{\exp}\left(-\tilde{\lambda}_{1}r^{2}t\text{ }_{2}F_{1}\left(1,1-\frac{2}{\alpha_{1}};2-\frac{2}{\alpha_{1}};-t\right) \right) }{t+1}\mathbf{f}_{R_{1}}(r) \text{d}r \text{d}t,
\end{align}
}

{\color{black}
\begin{align} \label{new_spec4}
&\mathcal{SE}^{(s)}_{d2} = \notag \\
& \int\limits_{0}^\infty \int\limits_{0}^\infty \frac{ \exp\left(-\tilde{\lambda}_{2}r^{2}t\text{ }_{2}F_{1}\left(1,1-\frac{2}{\alpha_{2}};2-\frac{2}{\alpha_{2}};-t\right)\right)}{t+1}  \mathbf{f}_{R_{2}}(r) \text{d}r \text{d}t, \end{align}
}
\normalsize

{\color{black}
\begin{align} \label{new_specd4}
&\mathcal{SE}^{(s)}_{c2} = \notag \\
& \int\limits_{0}^\infty \int\limits_{0}^\infty \frac{ \exp\left(-\tilde{\lambda}_{1}r^{2}t\text{ }_{2}F_{1}\left(1,1-\frac{2}{\alpha_{1}};2-\frac{2}{\alpha_{1}};-t\right)\right)}{t+1}  \mathbf{f}_{R_{c2}}(r) \text{d}r \text{d}t, \end{align}
}
\normalsize

\footnotesize
{\color{black}
\begin{align} \label{new_spec5}
&\mathcal{SE}^{(c)}_{B}=\mathcal{SE}^{(s)}_{dB} \notag \\
&=\int\limits_{0}^\infty \int\limits_{0}^\infty \frac{{\exp}\left(-\tilde{\lambda}_{2}r^{2}t\text{ }_{2}F_{1}\left(1,1-\frac{2}{\alpha_{2}};2-\frac{2}{\alpha_{2}};-t\right)\right)}{t+1}  \mathbf{f}_{R_{B}}(r) \mbox{d}r  \mbox{d}t,
\end{align}
}
\normalsize

\footnotesize
{\color{black}
\begin{align} \label{new_specc5}
\mathcal{SE}^{(s)}_{cB}=\int\limits_{0}^\infty \int\limits_{0}^\infty \frac{{\exp}\left(-\tilde{\lambda}_{1}r^{2}t\text{ }_{2}F_{1}\left(1,1-\frac{2}{\alpha_{1}};2-\frac{2}{\alpha_{1}};-t\right)\right)}{t+1}  \mathbf{f}_{R_{cB}}(r) \mbox{d}r  \mbox{d}t.
\end{align}
}
\normalsize

{As shown in equations \eqref{new_spec1}-\eqref{new_specc5}, two fold integrals are required to obtain the spectral efficiency for general path loss exponents, which is numerically complex to evaluate. For the special case of path loss exponents $\alpha_1=\alpha_2=4$}, the spectral efficiency  for all types of users can be evaluated via single integral as follows:

\footnotesize
\begin{align}
\mathcal{SE}^{(c)}_{1}=\int_{0}^\infty \frac{1}{t+1}\frac{\lambda_{1}+\lambda_{2}\sqrt{\tilde{P}_{21}}}{\lambda_{1}\rho(1,t)+ \lambda_{2}\sqrt{\tilde{P}_{21}}\rho(1,\frac{t}{B})}\text{d}t,
\end{align}
\begin{align}
\mathcal{SE}^{(c)}_{2}=\int_{0}^\infty\frac{1}{t+1}\frac{1}{\rho(1,t)}\text{d}t,
\end{align}
\normalsize
\footnotesize
\begin{multline}
\mathcal{SE}^{(c)}_{B}= \mathcal{SE}^{(s)}_{dB}=\int_{0}^\infty\frac{1}{t+1}\left(\frac{(\lambda_{1} \sqrt{\tilde{P}_{12}}+\lambda_{2})(\lambda_{1}\sqrt{P_{12}}+\lambda_{2})}{\lambda_{1}\left(\sqrt{P_{12}}-\sqrt{\tilde{P}_{12}}\right)}\right)\\ \left(\frac{\lambda_{1}\left(\sqrt{P_{12}}-\sqrt{\tilde{P}_{12}}\right)}{(\lambda_{2}\rho(1,t)+\lambda_{1}\sqrt{\tilde{P}_{12}})(\lambda_{2}\rho(1,t)+\lambda_{1}\sqrt{P_{12}})}\right)\mbox{d}t,
\end{multline}
\begin{align}
\mathcal{SE}^{(s)}_{1}=\int_{0}^\infty \frac{1}{t+1}\frac{\lambda_{1}+\lambda_{2}\sqrt{\tilde{P}_{21}}}{\lambda_{1}\rho(1,t)+\lambda_{2}\sqrt{\tilde{P}_{21}}}\text{d}t,
\end{align}
\begin{align}
\mathcal{SE}^{(s)}_{d2}=\int_{0}^\infty\frac{1}{t+1}\frac{\lambda_{2}+\lambda_{1}\sqrt{P_{12}}}{\lambda_{2}\rho(1,t)+\lambda_{1}\sqrt{P_{12}}} \text{d}t,
\end{align}
\begin{multline}
\mathcal{SE}^{(s)}_{c2}=\int\limits_0^\infty \frac{\bigg(1+\frac{\lambda_{1}}{\lambda_{2}}\sqrt{P_{12}} \bigg) \bigg( 1 - \frac{1}{1+\frac{\lambda_{2}}{\lambda_1} \sqrt{P_{21}}\rho(1,t)^{-1} }\bigg)}{(t+1)\rho(1,t)} dt,
\end{multline}
\begin{multline}
 \mathcal{SE}^{(s)}_{cB}=\int_{0}^\infty\frac{1}{t+1}\left(\frac{(\lambda_{1}\sqrt{\tilde{P}_{12}}  +\lambda_{2})(\lambda_{1}\sqrt{P_{12}}+\lambda_{2})}{\lambda_{1}\left(\sqrt{P_{12}}- \sqrt{\tilde{P}_{12}}\right)}\right)\\ \left( \frac{\lambda_{2}\left(\sqrt{\tilde{P}_{21}}-\sqrt{P_{21}}\right)}{\left( \lambda_{1}\rho(\sqrt{t},t)+\lambda_{2}\sqrt{P_{21}}\right)\left(\lambda_{1}\rho(\sqrt{t},t)+\lambda_{2}\sqrt{\tilde{P}_{21}}\right)} \right)\mbox{d}t.
\end{multline}
\normalsize

\section{Handover Analysis}
\label{mobility_analysis}

In this section, we take into account the effect of mobility on the system performance.  In order to compute the average throughputs in Eq. \eqref{watpu_c} and \eqref{watpu_v}, we need to compute the handover cost for both the CP/UP split network and conventional network architectures. The handover cost is a function of the handover rate per unit length of users trajectories, which is calculated in this section.

We assume that users move according to an arbitrary mobility pattern with velocity $\mathcal{V}$. The handover rate is determined based on the model obtained by {\color{black}Bao and Liang \cite{bao2015stochastic}, which gives the handover rate per unit length for arbitrary trajectories in a PPP multi-tier network.  Hence, the handover rate is independent of the underlying  mobility pattern}. Following \cite{bao2015stochastic}, the tier-$i$-to-tier-$j$ handover rate per unit length of an arbitrary trajectory is given by

\begin{equation}\label{HOC}
HO_{ij}^{(c)}= \frac{\lambda_{i}\lambda_{j}\mathcal{F}(x_{ij})}{\pi \left(\lambda_{i}+\lambda_{j} x_{ij}^{2}\right)^{\frac{3}{2}}},
\end{equation}
where $x_{12}=\left(\tilde{P}_{12}\right)^{\frac{1}{\alpha}}$, $x_{21}=\frac{1}{x_{12}}$, $x_{11}=x_{22}=1$, and
\begin{equation}
\mathcal{F}(x) = \frac{1}{x^{2}}\int_{0}^{\pi}\sqrt{(x^{2}+1)-2x\text{cos}(\theta)}d\theta.
\end{equation}

\noindent From Eq. (\ref{conv_cost}), the total handover cost per unit time in the conventional network is given by:
\footnotesize
\begin{equation}
D^{(c)}_{HO} =  \left(d^{(c)}\left(1-\mathcal{X}\right)+\tilde{d}^{(c)}\mathcal{X}\right) \frac{\mathcal{V}}{\pi}   \sum_{i=1}^2 \sum_{j=1}^2 \frac{\lambda_{i}\lambda_{j}\mathcal{F}(x_{ij})}{ \left(\lambda_{i}+\lambda_{j} x_{ij}^{2}\right)^{\frac{3}{2}}}.
\label{CHOD}
\end{equation}
\normalsize

In the CP/UP split network, an inter-anchor handover takes place when crossing a MBS-to-MBS cell boundary (see Fig.~\ref{weighted_Voronoi}); thus the inter-anchor handover rate is equivalent to the handover rate in the single-tier MBS case with density $\lambda_1$.   Following \cite{bao2015stochastic}, we calculate the inter-anchor handover rate per unit length in the CP/UP split architecture  network as follows:

\footnotesize
\begin{equation}\label{MHO}
MHO^{(s)}= \frac{4 \sqrt{\lambda_1}}{\pi}.
\end{equation}
\normalsize

\noindent As discussed earlier, intra-anchor handovers are defined as all types of handovers that do not require changing the anchor BS. Hence, the intra-anchor handover rate per unit length is given by:

\footnotesize
\begin{equation}\label{VHO}
VHO^{(s)}= \frac{2}{\pi}\left(\sum_{i=1}^2 \sum_{j=1}^2 \frac{\lambda_{i}\lambda_{j}\mathcal{F}(x_{ij})}{2 \left(\lambda_{i}+\lambda_{j} x_{ij}^{2}\right)^{\frac{3}{2}}} - 2 \sqrt{\lambda_1}\right).
\end{equation}
\normalsize
From Eq. (\ref{cpup_cost}), the total handover cost per unit time in the CP/UP split network is given by:

\footnotesize
\begin{align}
 \!\!\!\!\!\!\!\!D^{(s)}_{HO} &=  \frac{2 \mathcal{V} d^{(s)}_v }{\pi}\left(\sum_{i=1}^2 \sum_{j=1}^2 \frac{\lambda_{i}\lambda_{j}\mathcal{F}(x_{ij})}{2 \left(\lambda_{i}+\lambda_{j} x_{ij}^{2}\right)^{\frac{3}{2}}} - 2 \sqrt{\lambda_1}\right) \notag \\
 & \quad  \quad  \quad  \quad \quad  \quad+  \left(\left(1-\mathcal{Z}\right)d^{(s)}_m+\mathcal{Z}\tilde{d}_m^{(s)}\right)\mathcal{V}  \frac{4 \sqrt{\lambda_1}}{\pi} \notag \\
&= \frac{ \mathcal{V} }{\pi}  \Bigg( d^{(s)}_v \sum_{i=1}^2 \sum_{j=1}^2 \frac{\lambda_{i}\lambda_{j}\mathcal{F}(x_{ij})}{ \left(\lambda_{i}+\lambda_{j} x_{ij}^{2}\right)^{\frac{3}{2}}} \notag\\& \quad  \quad  \quad  \quad  \quad \quad  \quad + {4 \sqrt{\lambda_1}} \left(\left(1-\mathcal{Z}\right)d^{(s)}_m+\mathcal{Z}\tilde{d}_m^{(s)}- d^{(s)}_v\right)\Bigg).
\label{VHOD}
\end{align}
\normalsize
\noindent Note that the inter-anchor handover delay is equal to the conventional handover delay (i.e., $d^{(s)}_m = d^{(c)}$ and $\tilde{d}^{(s)}_m = \tilde{d}^{(c)}$) because the handover procedure is the same.  Thus we can infer from Eq. \eqref{CHOD} and \eqref{VHOD} that the handover cost depends on the relative values of $d^{(c)}$, $d^{(s)}_v$, $\lambda_2$, and $\lambda_1$.  In fact, in an ultra dense small cell network with $\lambda_2 >> \lambda_1$, we can obtain a bound on the maximum gain in terms of the handover cost that the CP/UP split architecture can offer when $\mathcal{X}=\mathcal{Z}=0$ as follows:
\begin{align} \label{huh}
{\mathcal{G}} &= \underset{ \lambda_2 \rightarrow \infty }{\lim} \frac{D^{(c)}_{HO}-D^{(s)}_{HO}}{D^{(c)}_{HO}}  = 1-\frac{d^{(s)}_v}{d^{(c)}}.
\end{align}

\noindent Note that the core network is mainly wired and the core network elements may be located far away from the network edge, and hence, core network signaling travels farther distances with lower speed\footnote{Wave propagation within any medium is less than the speed of electromagnetic waves in the air, which travels with the speed of light.}. Hence, the core network signaling may add significant delay to the handover procedure. For instance, if $d^{(c)} = 5 d^{(s)}_v$, Eq. \eqref{huh} shows that the CP/UP split architecture can offer $80\%$ reduction in the handover delay.  


\section{Model Validation and Numerical Results}
\label{validation_and_result}

In this section, we first validate our results via simulations using MATLAB. We then use the developed analytical model to compare the performance of the conventional and CP/UP split RAN architectures and obtain design insights.

Unless otherwise stated, we use the following parameters in our simulations and analysis.  The transmission powers are $P_{1}=50$ Watt and $P_{2}=5$ Watt. The bandwidth is $W=10$ MHz. The ABS factor is $\eta=0.3$. The percentage of control data in the available time/frequency resources is $\mu_{C}=0.3$ based on 3GPP Release 11 \cite{hoymann2013lean}. The biasing factor for the small BSs tier is $B=30$. The available air interface bandwidth for macro cells resource allocation is $W_{1}=2$ MHz, and for small cells resource allocation is $W_{2}=8$ MHz. We assume that the density of MBSs is $\lambda_{1}=2$ BS/km$^{2}$ and the density of mobile users is $\lambda^{u}=50$ users/km$^{2}$.   The path loss exponent is $\alpha_{1}=\alpha_{2}=4$.

{\color{black} \subsection{Model Validation}\label{Model_Validation}

\begin{figure*}[t!]
    \begin{center}
    \begin{subfigure}[t]{0.5\textwidth}
       \scalebox{0.36}[0.36]{\includegraphics{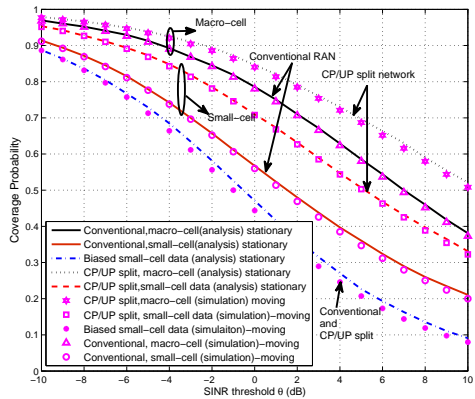}}\caption{ Conventional association \& CP/UP split data association.}
\label{coverage_probability_moving}
    \end{subfigure}%
    ~
    \begin{subfigure}[t]{0.5\textwidth}
       \scalebox{0.36}[0.36]{\includegraphics{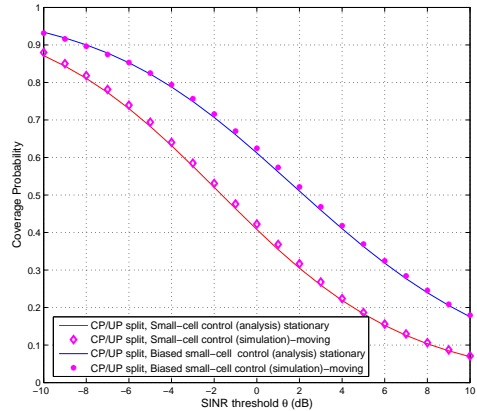}}\caption{ CP/UP split control association.}
\label{coverage_probability_moving_22}
    \end{subfigure}%
    ~
    \caption{Coverage probability as a function of the {\rm SINR} threshold $\theta$ for stationary users (analysis) and mobile users (simulation).}\label{validation}
\end{center}
\end{figure*}

 In each simulation run, the network is realized in $ 90\times90 \text{km}^{2}$ via two independent homogenous PPPs with densities $\lambda_1$ and $\lambda_2$. A test user is then generated at the origin and moves along five consecutive straight trajectories, each with a random length (Rayleigh distributed with parameter $1/\sqrt{2\pi\lambda_{1}}$) and random angle (uniformly distributed on $[0,2\pi]$). The trajectories are then partitioned with a 100-point resolution and the user's association and SINR are recorded at each point. The association type is determined based on Eq. \eqref{sets}. According to the association, the SINR value at each point of the test user trajectory is saved in one of the eight cumulative vectors corresponding to the 8 link types listed in Fig.~\ref{split_and_no_split}. Then, the above process is repeated 1000 times. The empirical CCDF of the values recorded in the eight cumulative vectors are then compared to the respective CCDF in Eq. \eqref{simulation_snir_1} to \eqref{simulation_snir_2}.

 Fig.~\ref{validation} plots the SINR CCDF obtained from the analysis (for stationary users) and the simulation (for mobile users). The figure shows that the analysis (for stationary users) closely captures the simulation result (for mobile users), confirming the validity of the proposed model for both stationary and mobile users. While the simulation result (for mobile users) considers the spatial correlation between SINR values across users' trajectories, the close match between the analysis and simulation result can be explained by the rapid spatial decay of the spatial correlation between the interference signal and the distance \cite{spatial_martin}. Hence, we can deduce that averaging over all locations in all network configurations closely captures the averaging over all trajectories in all network configurations. Fig.~\ref{coverage_probability_moving} shows that the CP/UP split architecture offers higher coverage probability, for tier-1 and tier-2 users data links, than the conventional RAN architecture due to the absence of cross-tier interference. Fig. \ref{coverage_probability_moving_22} shows that SINR coverage probability for the control signaling of the biased SBS users is better than that of the unbiased SBSs users. This is because the biased SBSs users are closer on average to the MBSs than the unbiased SBSs users.




\subsection{Handover Rate and Throughput}

Fig. \ref{simulation_analysis} visualizes the handover rates per unit length of an arbitrary trajectory as a function of the density of SBS in the conventional and CP/UP split networks. The graph shows that small cell densification linearly increases the total number of handovers in the conventional network architecture. Looking into the explicit handover types, we notice that $HO_{11}^{(c)}$ (handovers between MBSs) decreases as the density of small cells increases. The reason is that the boundaries of MBSs become more populated by SBSs when $\lambda_2$ increases. Hence, a MBS-to-MBS handover is replaced by MBS-to-SBS followed by SBS-to-MBS handover and possibly several SBS-to-SBS handovers inbetween. Also, $HO_{22}^{(c)}$ linearly increases and  $HO_{12}^{(c)} =  HO_{21}^{(c)}$  saturates. Hence, with high SBS densities, the handover rate is dominated by $HO_{22}^{(c)}$, which motivates the anchoring solution via CP/UP splitting to reduce the handover delay. As shown in the graph, in the CP/UP split network architecture, the inter-anchor handover rate is kept constant due to the constant density of the MBSs. However, the intra-anchor handover rate increases linearly with $\lambda_2$.

\begin{figure}[!t]
\begin{center}
\scalebox{0.35}[0.35]{\includegraphics{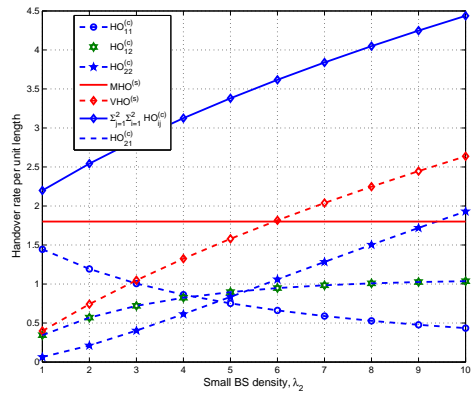}}
\end{center}
\caption{Handover rate per unit length}
\vspace{-.2cm}
 \label{simulation_analysis}
\end{figure}


\begin{figure*}[t!]
    \begin{center}
    \begin{subfigure}[t]{0.45\textwidth}
       \scalebox{0.36}[0.36]{\includegraphics{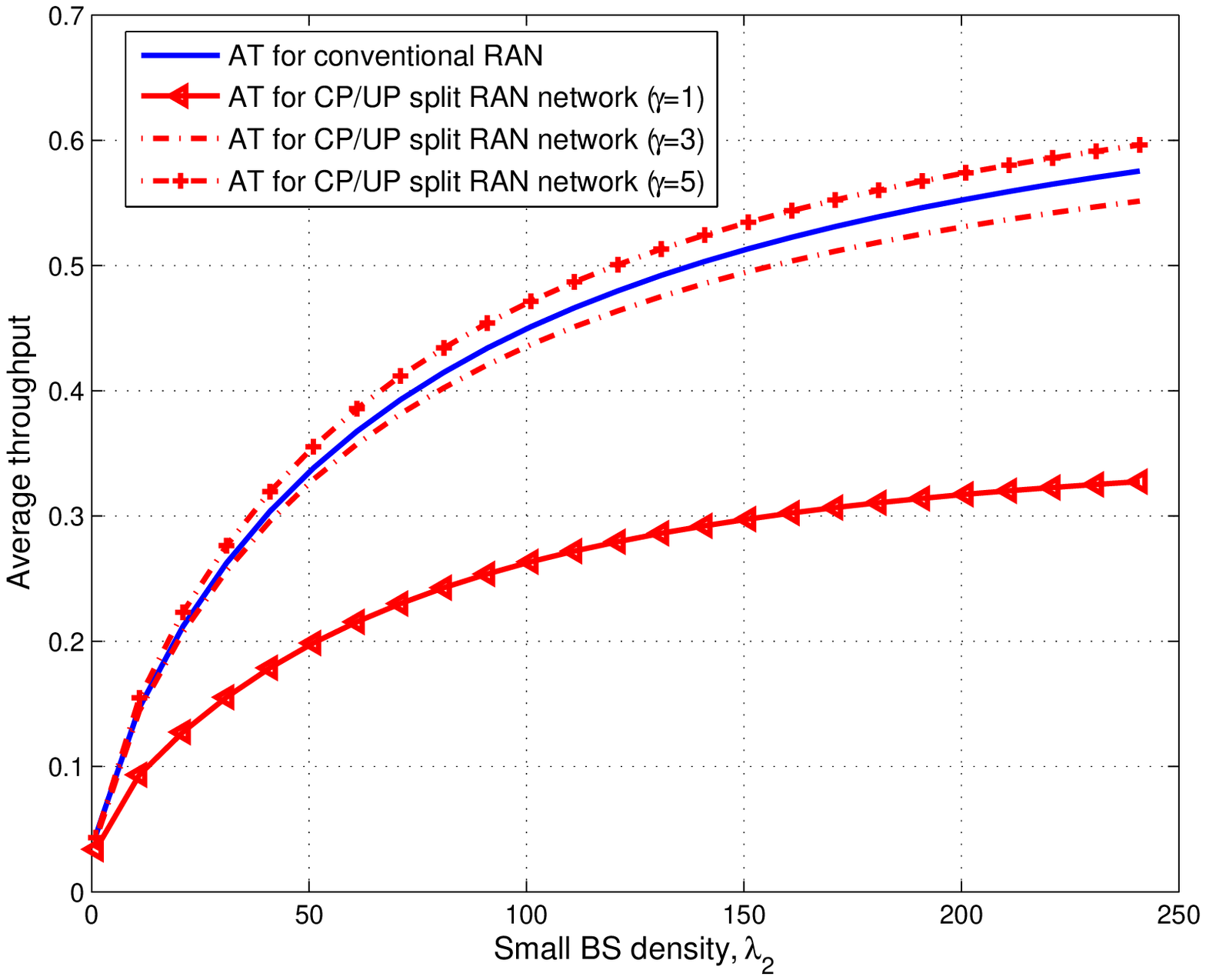}}\caption{ Stationary $\mathcal{V}=0 \text{ } km/h$.}
\label{average_rate_three_cases}
    \end{subfigure}%
    ~
    \begin{subfigure}[t]{0.45\textwidth}
       \scalebox{0.36}[0.36]{\includegraphics{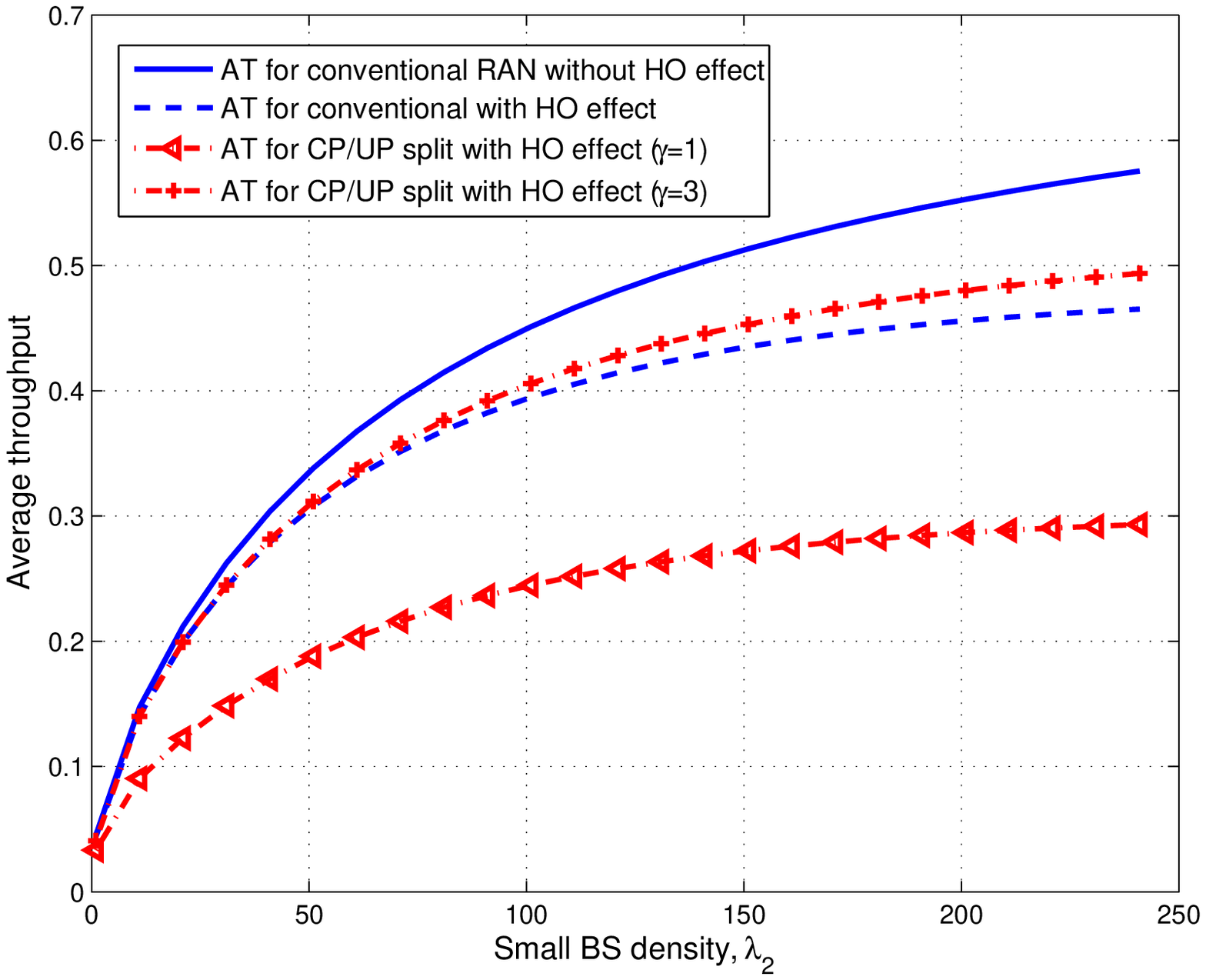}}\caption{Low speed $\mathcal{V}=50 \text{ } km/h$.}
\label{Effect_ATPU}
    \end{subfigure}%
    ~

    \begin{subfigure}[t]{0.45\textwidth}
       \scalebox{0.36}[0.36]{\includegraphics{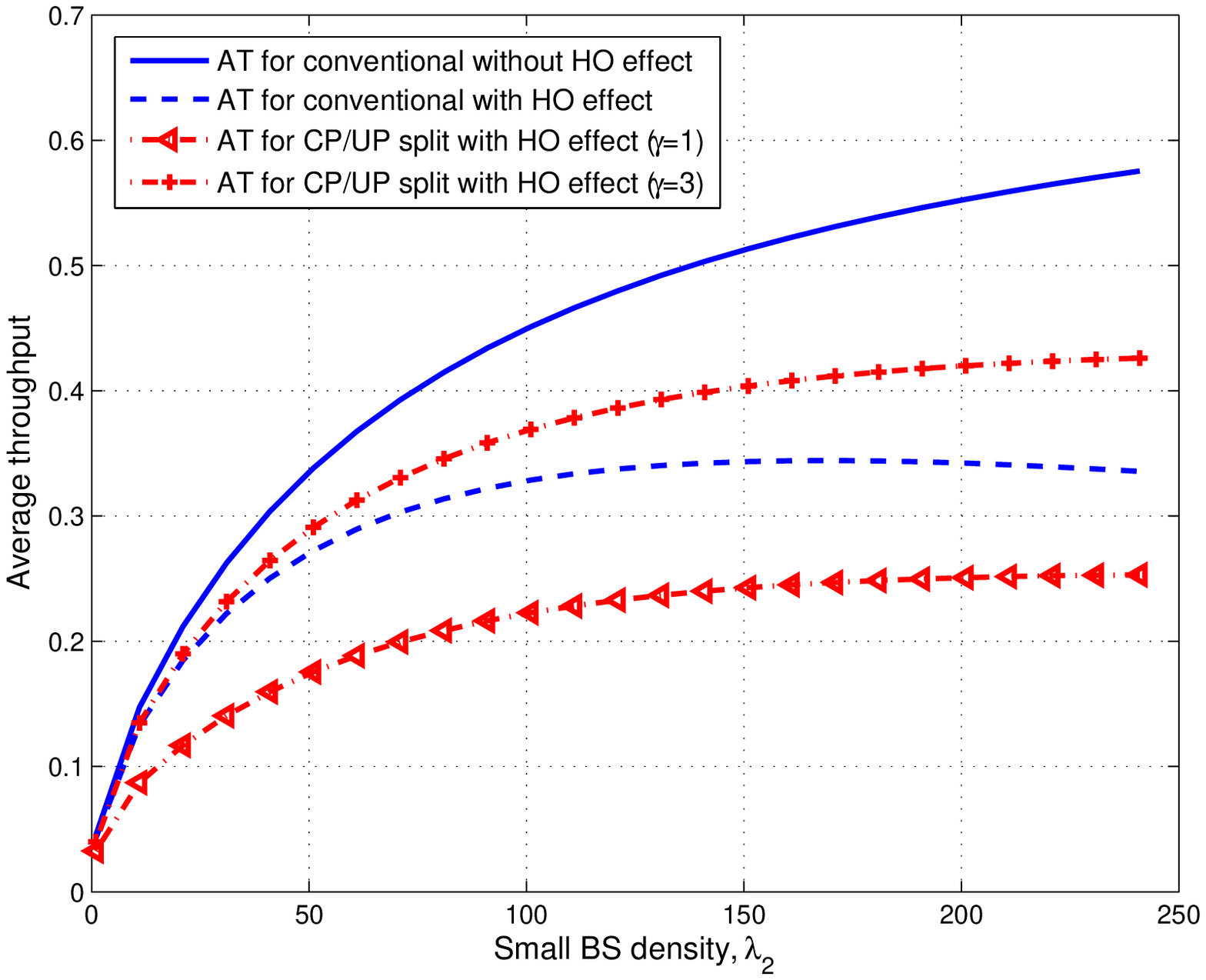}}\caption{Medium speed $\mathcal{V}=108 \text{ } km/h$.}
\label{Effect_ATPU_2}
    \end{subfigure}
        ~
    \begin{subfigure}[t]{0.45\textwidth}
       \scalebox{0.36}[0.36]{\includegraphics{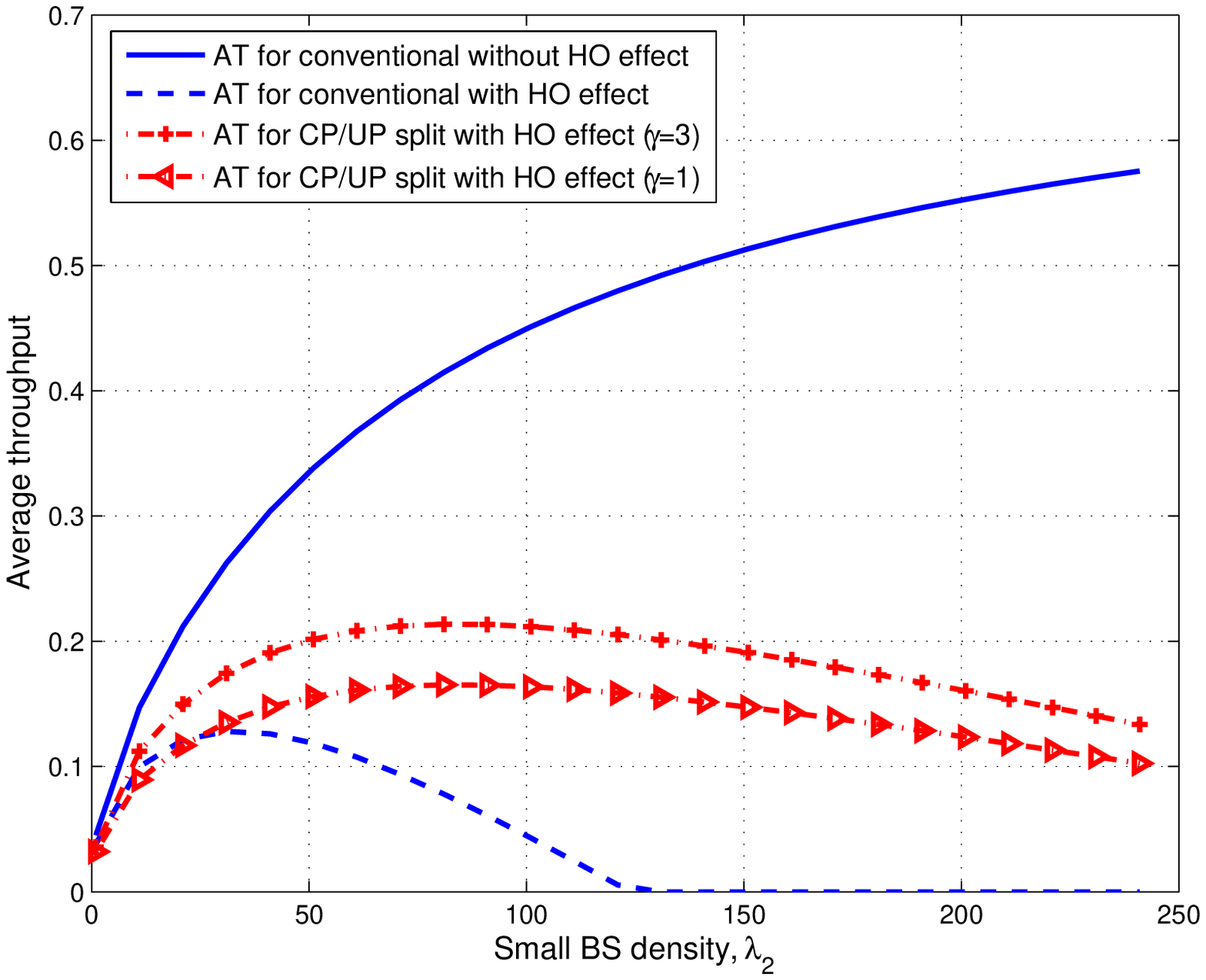}}\caption{High speed $\mathcal{V}=360 \text{ } km/h$.}
\label{Effect_ATPU_3}
    \end{subfigure}
    \caption{Average throughput with and without handover cost for mobile user with different velocities for $\gamma \in\{1, 3, 5\}$ and $\mathcal{X}=\mathcal{Z}=0$.}\label{E_ATPU}
\end{center}
\end{figure*}

The next set of simulation results show the effect of mobility, control signaling reduction factor $\gamma$, the availability of X2 interface between BSs, and SBS density on the average user throughput. Unless otherwise stated, we assume that $d^{(c)}=0.7$ seconds and that $d^{(s)}_{v}=\tilde{d}^{(c)}=\tilde{d}_m^{(s)}=0.5d^{(c)}$~\cite{mahmoodiusing}.  Fig.~\ref{E_ATPU} shows the effect of  the handover delay on the average user throughput in the conventional and CP/UP split architectures for different mobility profiles: (a) stationary $\mathcal{V}=0\text{ } km/h$, (b) low velocity $\mathcal{V}=50\text{ } km/h$ (e.g., driving in the city), (c) medium velocity $\mathcal{V}=108\text{ } km/h$ (e.g., traveling on highways or in monorails in city downtowns), and (d) high velocity $\mathcal{V}=360\text{ } km/h$. (e.g., traveling on a high speed train {\color{black}such as Shinkansen when passing through downtown Tokyo, Japan}). 

In the case of stationary users, Fig.~\ref{average_rate_three_cases} shows that a high control reduction factor $\gamma$ is required for the CP/UP split architecture to achieve an equivalent average throughout to the conventional network architecture. This result can be interpreted by the poor control rate provided by MBSs to the unbiased phantom cell users when compared to the rate they get from the SBSs (cf. Fig. \ref{coverage_probability_moving_22}). Hence, offloading the control signaling to the MBSs requires a high control reduction factor to compensate for such rate loss. Note that the per user rate for unbiased users of the SBSs increases with $\lambda_2$, and hence, offloading control to the MBSs incurs higher rate loss. Consequently, the CP/UP split architecture is not beneficial to networks with stationary users unless a high control reduction factor can be achieved.




For mobile users, Figs. \ref{Effect_ATPU}, \ref{Effect_ATPU_2} and \ref{Effect_ATPU_3} show that the CP/UP split architecture is beneficial especially for high speeds and $\mathcal{X}=\mathcal{Z}=0$; i.e., there is no X2 interface handovers on both conventional and CP/UP split architectures. Note that we show the ideal case; i.e, AT for stationary users, to clearly visualize the effect of mobility on the average throughput. Figs. \ref{Effect_ATPU} and \ref{Effect_ATPU_2} show that a control reduction factor of $\gamma =3$ is sufficient for the CP/UP split architecture to outperform the conventional network architecture when users move at low or medium speeds.  When the mobility speed is high (Fig. \ref{Effect_ATPU_3}), the CP/UP split network outperforms the conventional network even without control reduction (i.e., $\gamma=1$).  More importantly, only the CP/UP split network can support users moving at such high speeds while the conventional network cannot.

It is important to note that Fig. \ref{E_ATPU} is plotted for $d^{(s)}_{v}=0.5d^{(c)}$. The CP/UP split architecture can offer even higher throughput gains if the intra-anchor delay is lowered.  Fig. \ref{Effect_ATPU_4} shows the additional gain that the CP/UP split network offers when $d^{(s)}_{v}=0.3d^{(c)}$ versus the case where $d^{(s)}_{v}=0.5d^{(c)}$.  The graph demonstrates the importance of lowering the intra-anchor delay and minimizing the involvement of the core network during handovers. Therefore, the CP/UP split architecture can be used to increase the throughput of mobile users in dense small cell deployments by making the MBSs act as handover anchors instead of involving the core network in handovers.

 Fig. \ref{AT_XZ} shows the effect of the direct X2 interface availability between BSs on the average throughput in the conventional and CP/UP splitting architectures. The figure shows that the X2 interfaces have more prominent effect on the conventional network architecture because it reduces the delay for all handover types. On the other hand, the X2 interference does not have a noticeable effect on in CP/UP split architecture because it only reduces the inter-anchor handover delay, which is considered a rare handover event.  The figure also shows that the relative performance gains between the conventional and CP/UP splitting architectures highly depends on the X2 interface availability. Particularly, there are critical points at which the conventional network with sufficient X2 interface deployment outperforms the CP/UP split architecture in terms of average throughput. Such critical points are depicted in Fig. \ref{AT_XZ} at  $\mathcal{X}=\mathcal{Z} = 0.5$, $\mathcal{X}=\mathcal{Z}= 0.8$, and $\mathcal{X}=\mathcal{Z}= 0.95$ for $\mathcal{V} = 50 \text{ } km/h$, $\mathcal{V} = 108 \text{ } km/h$. and $\mathcal{V} = 360 \text{ } km/h$, respectively.

\begin{figure}[!t]
\begin{center}
\scalebox{0.37}[0.37]{\includegraphics{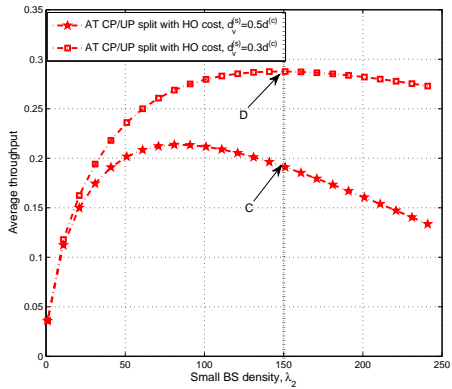}}
\end{center}
\caption{Average throughput with handover cost and intra-anchor handover delay values $d^{(s)}_{v}=0.5d^{(c)}$ and $d^{(s)}_{v}=0.3d^{(c)}$ ($\mathcal{V} = 360 \text{ } km/h$, $\gamma=3$, and $\mathcal{X}=\mathcal{Z}=0$).}
\vspace{-.2cm}
\label{Effect_ATPU_4}
\end{figure}


\begin{figure}[!h]
\begin{center}
\scalebox{0.35}[0.35]{\includegraphics{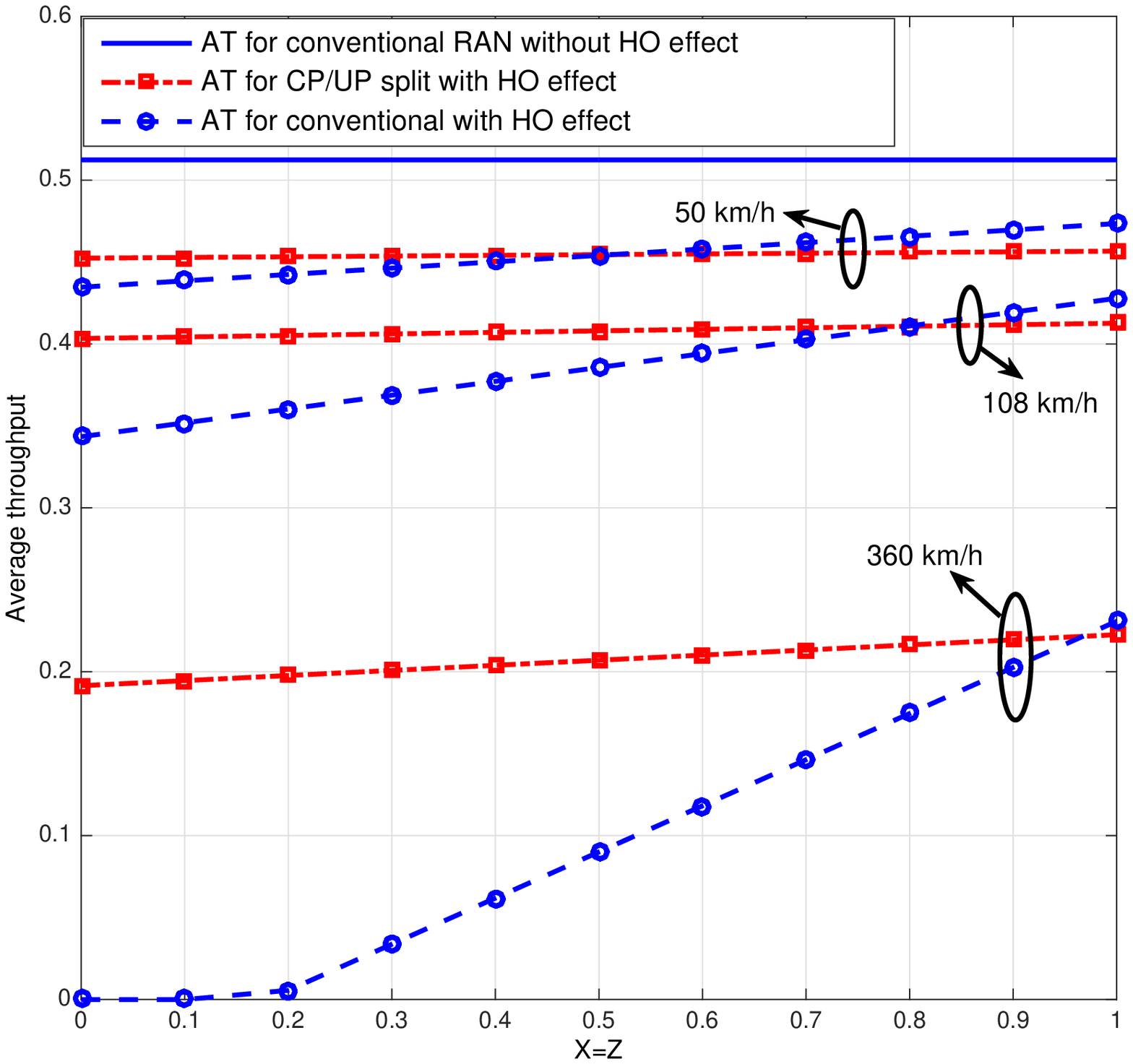}}
\end{center}
\caption{Average throughput with and without handover cost when $\gamma=3$, $\lambda_{2}=150$ BS/km$^{2}$, and $d^{(s)}_{v}=\tilde{d}^{(c)}=\tilde{d}_m^{(s)}=0.5d^{(c)}$. }
\vspace{-.2cm}
\label{AT_XZ}
\end{figure}


\subsection{Feasibility of the CP/UP Split Architecture }

\begin{figure}[!t]
\begin{center}
\scalebox{0.34}[0.34]{\includegraphics{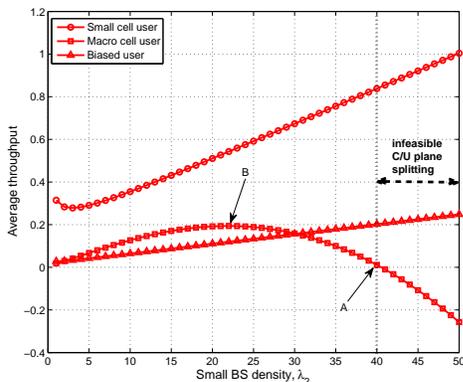}}
\end{center}
\vspace{-20pt}
\caption{Average throughput of small cell, macro cell and  biased users as a function of small cell density  ($\mathcal{V}=0$, $\gamma$ = 3 and $B=30$)}
\vspace{-.2cm}
\label{weighted_average}
\end{figure}

To examine the feasibility of the CP/UP split architecture as stated in Corollary~\ref{col}, we plot Fig.~\ref{weighted_average}, which shows the average throughputs for all types of users as functions of the SBS density, assuming  $\mathcal{V}=0$, $\gamma=3$, and $B=30$. Note that we assume saturation conditions such that newly added SBSs always have users to serve.  The graph shows the breaking point (point A in Fig.~\ref{weighted_average}) at which the MBSs fail to provide the control signaling required by phantom cell users. Point A is the point at which the inequality \eqref{condition_col} is violated. Note that the CP/UP split architecture can still be made feasible by allocating more spectrum to the MBSs or enhancing the control reduction factor $\gamma$ as shown in Corollary~\ref{col}.

\subsection{Design Insights}
From the above numerical results, several design insights can be drawn for the  CP/UP split network architecture. First,  the CP/UP architecture becomes more appealing for higher mobility profiles when the availability of direct X2 interface between the BSs is low, in which the control signaling reduction factor plays a key role in the throughput gains when compared to the conventional architecture. The amount of delay reduction provided by the intra-anchor handover also has a significant impact on the throughput gains provided by the CP/UP split networks. For instance, Fig. \ref{Effect_ATPU_4} shows a 60\% throughput improvement when the intra-anchor handover delay  $d^{(s)}_{v}$ is reduced from $0.5d^{(c)}$ (point C in Fig. \ref{Effect_ATPU_4}) to $0.3d^{(c)}$ (point D in Fig. \ref{Effect_ATPU_4}).

Cellular operators can solve the excess handover problem, which is coupled with network densification, either by deploying more X2 interfaces between adjacent BSs or applying CP/UP splitting. While the former reduces the handover delay only, the latter reduces both the handover delay as well as the signaling overhead. {\color{black} Note that at higher SBSs densities and/or control reduction factors, the conventional network may not achieve the CP/UP split throughput even with $100\%$ X2 deployment}. Consequently, the CP/UP split architecture is more appealing in ultra dense environments with high mobility profiles.

Another noteworthy insight is that there is a tradeoff between traffic offloading via biasing and control offloading via the CP/UP architecture on macrocell users rate. As shown in Fig. \ref{weighted_average}, there is a turning point for the average throughput of macro cell users at $\lambda_{2} = 22$ (point B in Fig. \ref{weighted_average}).  For the given network configuration, prior to point B, the positive impact of offloading users traffic to phantom cells (i.e., decreasing $\mathcal{N}_1$)  dominates the negative impact of offloading control signaling from the phantom cells to MBSs. Then, the situation is reversed after point B and the negative impact of the control burden dominates the positive impact of traffic offloading until the infeasibility  point is reached (point A). Such tradeoff can be used to optimize the biasing factor such that the macrocell users rate is maximized.   
\section{Conclusion}
\label{conc_future}

We present a novel mobility-aware analytical paradigm for CP/UP split RAN network architecture with flexible user association. We derive tractable mathematical expressions for coverage probability and user throughput, which can be reduced to closed-form expressions in special cases.  The analysis takes into account the control signaling overhead, spectrum allocation schemes, interference coordination via almost blank subframes, the availability of X2 interface between BSs, and delay incurred by handovers.   We then use the developed model to quantify the performance gains offered by the CP/UP split RAN network architecture.  In particular, we quantify the impacts of handover delay and mobility speed on the user throughput.  We also examine the effects of small cell density, control reduction factor, and core network delay on the user throughput.

The developed model shows that the handovers impose a fundamental limit on the performance gain that can be obtained via densification. In moderate and high mobility profiles, the CP/UP split network architecture offers a potential solution to reduce the control overload and mitigate the handover delay, and hence, improve the network densification gain in networks with low availability of direct X2 interface
between BSs. It is also crucial to know the optimal small cell density for a specific network configuration in order to balance the trade-off between the offloading of user data traffic away from MBSs and control signaling towards MBSs in order to maximize the network throughput.


\appendices

\section{Proof of Lemma~\ref{rate_lemma}} \label{rate_lemma_proof}

The phantom cells dedicate $\eta$ and $(1-\eta)$ fraction of the time to serve biased and non-biased users, respectively.  Consider a time interval of $t$ seconds. Then the numbers of data bits sent by each phantom BS to non-biased and biased users are $ (1-\eta) t \mathcal{R}_{d2}$ bits and $ \eta t \mathcal{R}_{dB}$ bits, respectively. On average, there are $\frac{\lambda_2}{\lambda_1}$ phantom BSs per MBS, and hence, the MBS should be able to convey control signaling amounts of  $ \frac{\lambda_2}{\lambda_1} \frac{\mu_{C} (1-\eta) t \mathcal{R}_{d2}}{\gamma}$ bits and $ \frac{\lambda_2}{\lambda_1} \frac{\mu_{C} \eta t \mathcal{R}_{dB}}{\gamma}$ bits to non-biased and biased users, respectively, during time interval $t$. However, the MBS sends the control bits with the rates of $\mathcal{R}_{c2}$ and $\mathcal{R}_{cB}$  for non-biased and biased phantom cell users, respectively. Hence, the amount of time required to send the control signaling by the MBS is $\frac{\lambda_2}{\lambda_1} \frac{\mu_{C} (1-\eta) t \mathcal{R}_{d2}}{\gamma \mathcal{R}_{c2}}$ seconds  and $\frac{\lambda_2}{\lambda_1} \frac{\mu_{C} \eta  t \mathcal{R}_{dB}}{\gamma \mathcal{R}_{cB}}$ seconds for non-biased and biased phantom cell users, respectively. Consequently, the remaining time for the MBS to serve macrocell users is $\left(t- \frac{\lambda_2}{\lambda_1} \frac{\mu_{C} (1-\eta) t \mathcal{R}_{d2}}{\gamma \mathcal{R}_{c2}} - \frac{\lambda_2}{\lambda_1} \frac{\mu_{C} \eta t \mathcal{R}_{dB}}{\gamma \mathcal{R}_{cB}} \right)$ seconds. Hence, the average number of bits the MBS conveys to macrocell users during time interval $t$ is $\mathcal{R}^{(s)}_1 \left(t- \frac{\lambda_2}{\lambda_1} \frac{\mu_{C} (1-\eta) t \mathcal{R}_{d2}}{\gamma \mathcal{R}_{c2}} - \frac{\lambda_2}{\lambda_1} \frac{\mu_{C} \eta t \mathcal{R}_{dB}}{\gamma \mathcal{R}_{cB}}\right) $ bits. Dividing the above expression by $t$, we obtain the average rate at which data is delivered to macrocell users as  $\mathcal{R}^{(s)}_1  \left(1- \frac{\lambda_2}{\lambda_1} \frac{\mu_{C} (1-\eta) \mathcal{R}_{d2}}{\gamma \mathcal{R}_{c2}} - \frac{\lambda_2}{\lambda_1} \frac{\mu_{C} \eta \mathcal{R}_{dB}}{\gamma \mathcal{R}_{cB}}\right) $. Then \eqref{macro_rate_lemma} is obtained by replacing $(1-\eta) \mathcal{R}_{d2}$ by $\mathcal{T}^{(s)}_2$, replacing $\eta \mathcal{R}_{dB}$ by $\mathcal{T}^{(s)}_B$ and multiplying the above expression by ($1-\mu_{C}$).

\normalsize

%

\section{Proof of Lemma~\ref{lem_distances}} \label{distances}
From the independence of the PPPs of the macro and phantom BSs , the joint pdf of the distances between a generic user and his nearest phantom BS and nearest MBS is given by $f_{r_1,r_2}(x,y)= 4 \pi^2 x y\lambda_{1}\lambda_{2} e^{-\pi (\lambda_1 x^2 +\lambda_2 y^2)}$, $x,y >0$. The control link distributions are given by
\footnotesize
\begin{align}
\mathbb{P}\left\{ r_1 < x \vert \mathcal{\Large\emph{u}}_2\right\} &= \frac{\mathbb{P}\left\{r_1 < x , \mathcal{\Large\emph{u}}_2\right\} }{\mathbb{P}\left\{ \mathcal{\Large\emph{u}}_2\right\} } = \frac{\mathbb{P}\left\{r_1 < x , {P_1 r_1^{-\alpha_1}}<{P_2 r_2^{-\alpha_1}}\right\} }{\mathbb{P}\left\{ {P_1 r_1^{-\alpha_1}}<{P_2 r_2^{-\alpha_1}}\right\} } \notag \\
&= \frac{\mathbb{P}\left\{r_1 < x , r_2< \left(\frac{P_2}{P_1}\right)^\frac{1}{\alpha_2} r_1^\frac{\alpha_1}{\alpha_2}\right\} }{\mathcal{A}_2 }.
\end{align}
\normalsize
Hence, the pdf of $R_{2c}$ is given by

\small
\begin{align}
f_{R_{c2}}(x) &= \frac{1}{\mathcal{A}_2}\frac{{\rm d}\mathbb{P}\left\{r_1 < x , r_2< \left(\frac{P_2}{P_1}\right)^\frac{1}{\alpha_2} r_1^\frac{\alpha_1}{\alpha_2}\right\} }{{\rm d}x}
\notag \\
&= \frac{1}{\mathcal{A}_2}\int_{0}^{\left(\frac{P_2}{P_1}\right){^\frac{1}{\alpha_2}}  x^{\frac{\alpha_1}{\alpha_2}}} f_{r_1,r_2}(x,y) dy.
\end{align}
\normalsize
Similarly, the pdf of $R_{cB}$ is derived as follows.
\small
\begin{align}
& \!\!\!\!\!\!\!\!f_{R_{cB}}(x) = \frac{1}{\mathcal{A}_B}\int_{\left(\frac{P_2}{P_1}\right){^\frac{1}{\alpha_2}} x^{\frac{\alpha_1}{\alpha_2}}}^{\left(\frac{B P_2}{P_1}\right){^\frac{1}{\alpha_2}}  x^{\frac{\alpha_1}{\alpha_2}}} f_{r_1,r_2}(x,y) dy.
\end{align}
\normalsize




\section{Proof of Lemma \ref{average_trans_rate_three_types}}
\label{Average transmission rate_prof}

The coverage probability, which is the {\em ccdf} of the {\rm SINR}, can be expressed in terms of the Laplace transform (LT) of the aggregate interference. Using the general SINR model in \eqref{gen_sinr}, the coverage probability is given by

\footnotesize
\begin{align}
\mathbb{P}[{\rm SINR}>\theta]&=\mathbb{P}\left[\frac {P_{BS}H {r_o}^{-\alpha}}{\mathcal{I}_{agg}+\sigma^2}>\theta\right]\notag\\
&\stackrel{(a)}{=}\int_{0}^\infty\mbox{exp}\left(-\frac{\sigma^{2} \theta r_{o}^{\alpha}}{P_{BS}} \right)\mathcal{L}_{\mathcal{I}_{agg}}(\theta \frac{r_{o}^{\alpha}}{P_{BS}})f_{r_o}(r),\label{atr_1}
\end{align}
\normalsize
where (a) follows from the exponential distribution of $H$ and the definition of the LT \cite{elsawy2013survey, dhillon2012modeling}, and the parameters in \eqref{atr_1} can be obtained for each user case from Table~\ref{SINR_Table}. In the case of shared spectrum access, $\mathcal{I}_{agg}$ is the superposition of two independent interferences from the two tiers $1$ and $2$, and hence, can be decomposed to the multiplications of the LTs of the interferences from each tier as

\begin{align}
\mathcal{L}_{\mathcal{I}_{agg}}(s) = \mathcal{L}_{\mathcal{I}_{1}}(s)  \mathcal{L}_{\mathcal{I}_{2}}(s).
\label{ltss}
\end{align}

\noindent In the other cases, the aggregate interference is simply single-tier interference. The LT of the interference for a given network tier $k \in \{1,2\}$ is calculated as:

\footnotesize
\begin{align}
\mathcal{L}_{\mathcal{I}_k}(s)&=\mathbb{E}\left[e^{-sP_{k}\sum\limits_{x\in\mathbf{\Phi_{k}}\backslash b_o}H_{x} x^{-\alpha_{k}}}\right] =\mathbb{E}\left[\prod_{x\in\mathbf{\Phi_{1}}\backslash b_o} e^{-sP_{k}H_{x} x^{-\alpha_{k}}}\right]\notag\\
&\stackrel{(b)}{=}\mbox{exp}\Bigg\{-2\pi \lambda_{k}\int_{\|b_o\|}^\infty \frac{sP_{k} x}{{x^{\alpha_{k}}}+sP_{k}} \mbox{d}x\Bigg\}.
\label{laplace_i1x}
\end{align}
\normalsize

\noindent where (b) follows by from the probability generating functional of the PPP and the i.i.d. exponential distribution of $H_x$, and $b_o$ is the location of the serving BS determined by the employed association criterion. The lemma is obtained by calculating the LT of the aggregate interference affecting the test user according to  Table~\ref{SINR_Table} using \eqref{ltss} and \eqref{laplace_i1x}, in which the  the location of the serving BS $b_o$ is obtained via the association criterion given in \eqref{sets}. Then, by substituting the LT of the interference in \eqref{atr_1} and integrating over the appropriate link distance given in Table~\ref{SINR_Table}, we obtain the coverage probabilities.

\ifCLASSOPTIONcaptionsoff
  \newpage
\fi



%

\bibliographystyle{IEEEtran}
\bibliography{phantom}

\begin{thebibliography}{10}
\providecommand{\url}[1]{#1}
\csname url@samestyle\endcsname
\providecommand{\newblock}{\relax}
\providecommand{\bibinfo}[2]{#2}
\providecommand{\BIBentrySTDinterwordspacing}{\spaceskip=0pt\relax}
\providecommand{\BIBentryALTinterwordstretchfactor}{4}
\providecommand{\BIBentryALTinterwordspacing}{\spaceskip=\fontdimen2\font plus
\BIBentryALTinterwordstretchfactor\fontdimen3\font minus
  \fontdimen4\font\relax}
\providecommand{\BIBforeignlanguage}[2]{{%
\expandafter\ifx\csname l@#1\endcsname\relax
\typeout{** WARNING: IEEEtran.bst: No hyphenation pattern has been}%
\typeout{** loaded for the language `#1'. Using the pattern for}%
\typeout{** the default language instead.}%
\else
\language=\csname l@#1\endcsname
\fi
#2}}
\providecommand{\BIBdecl}{\relax}
\BIBdecl

\bibitem{Andrews_5G}
J.~G. Andrews, S.~Buzzi, W.~Choi, S.~V. Hanly, A.~Lozano, A.~C. Soong, and
  J.~C. Zhang, ``What will 5{G} be?'' \emph{IEEE Journal on Selected Areas in
  Communications}, vol.~32, no.~6, pp. 1065--1082, 2014.

\bibitem{split2012phantom}
H.~Ishii, Y.~Kishiyama, and H.~Takahashi, ``A novel architecture for {LTE-B}:
  {C}-plane/{U}-plane split and {P}hantom {C}ell concept,'' in \emph{IEEE
  Globecom Workshops}, Anaheim, CA, USA, December 2012, pp. 624--630.

\bibitem{hoymann2013lean}
C.~Hoymann, D.~Larsson, H.~Koorapaty, and J.-F. Cheng, ``A lean carrier for
  {LTE},'' \emph{IEEE Communications Magazine}, vol.~51, no.~2, pp. 74--80,
  2013.

\bibitem{elsawy2013survey}
H.~ElSawy, E.~Hossain, and M.~Haenggi, ``Stochastic geometry for modeling,
  analysis, and design of multi-tier and cognitive cellular wireless networks:
  A survey,'' \emph{IEEE Communications Surveys \& Tutorials}, vol.~15, no.~3,
  pp. 996--1019, 2013.

\bibitem{Sarabjot2013partition}
S.~Singh and J.~Andrews, ``Joint resource partitioning and offloading in
  heterogeneous cellular networks,'' \emph{IEEE Transaction on Wireless
  Communication}, vol.~13, pp. 888--901, 2013.

\bibitem{jo2011outage}
H.-S. Jo, Y.~J. Sang, P.~Xia, and J.~G. Andrews, ``Outage probability for
  heterogeneous cellular networks with biased cell association,'' in
  \emph{Proc. of the IEEE GLOBECOM}, Houston, Texas, USA, December 2011, pp.
  1--5, 2011.

\bibitem{jo2012heterogeneous}
------, ``Heterogeneous cellular networks with flexible cell association: A
  comprehensive downlink {SINR} analysis,'' \emph{IEEE Transactions on Wireless
  Communications}, vol.~11, no.~10, pp. 3484--3495, 2012.

\bibitem{elsawy2014uplink}
H.~ElSawy and E.~Hossain, ``On stochastic geometry modeling of cellular uplink
  transmission with truncated channel inversion power control,'' \emph{IEEE
  Transactions on Wireless Communications}, vol.~13, no.~8, pp. 4454--4469,
  August 2014.

\bibitem{dhillon2012uplink}
H.~S. Dhillon, T.~D. Novlan, and J.~G. Andrews, ``Coverage probability of
  uplink cellular networks,'' in \emph{IEEE Global Communications Conference
  (GLOBECOM)}, Anaheim, California, USA, 2012, pp. 2179--2184.

\bibitem{mimo2013Di_Renzo}
M.~Di~Renzo and P.~Guan, ``A mathematical framework to the computation of the
  error probability of downlink {MIMO} cellular networks by using stochastic
  geometry,'' \emph{IEEE Transaction on Wireless Communications}, vol.~62,
  no.~8, pp. 2860--2879, August 2014.

\bibitem{modeling_MIMO_ASE_harpreet}
H.~Dhillon, M.~Kountouris, and J.~Andrews, ``Downlink {MIMO} {H}et{N}ets:
  {M}odeling, ordering results and performance analysis,'' \emph{IEEE
  Transactions on Wireless Communications}, vol.~12, no.~10, pp. 5208--5222,
  October 2013.

\bibitem{cao2012optimal}
D.~Cao, S.~Zhou, and Z.~Niu, ``Optimal base station density for
  energy-efficient heterogeneous cellular networks,'' in \emph{IEEE
  International Conference on Communications (ICC)}, 2012, pp. 4379--4383.

\bibitem{mukherjee2013energy}
S.~Mukherjee and H.~Ishii, ``Energy efficiency in the phantom cell enhanced
  local area architecture,'' in \emph{IEEE Wireless Communications and
  Networking Conference (WCNC)}, 2013, pp. 1267--1272.

\bibitem{zzz}
S.~Zhang, J.~Gong, S.~Zhou, and Z.~Niu, ``How many small cells can be turned
  off via vertical offloading under a separation architecture?'' \emph{IEEE
  Transactions on Wireless Communications}, vol.~14, no.~10, pp. 5440--5453,
  2015.

\bibitem{Hazem2015}
H.~Ibrahim, H.~ElSawy, U.~T. Nguyen, and M.-S. Alouini, ``Modeling virtualized
  downlink cellular networks with ultra-dense small cells,'' in \emph{IEEE
  International Conference on Communications (ICC)}, 2015, pp. 5360--5366.

\bibitem{lin2013towards}
X.~Lin, R.~K. Ganti, P.~J. Fleming, and J.~G. Andrews, ``Towards understanding
  the fundamentals of mobility in cellular networks,'' \emph{IEEE Transactions
  on Wireless Communications}, vol.~12, no.~4, pp. 1686--1698, 2013.

\bibitem{bao2015stochastic}
W.~Bao and B.~Liang, ``Stochastic geometric analysis of user mobility in
  heterogeneous wireless networks,'' \emph{IEEE Journal on Selected Areas in
  Communications}, 2015.

\bibitem{sadr2015handoff}
S.~Sadr and R.~Adve, ``Handoff rate and coverage analysis in multi-tier
  heterogeneous networks,'' \emph{IEEE Transactions on Wireless
  Communications}, vol.~14, no.~5, pp. 2626 -- 2638, 2015.

\bibitem{zhangdelay}
G.~Zhang, T.~Q. Quek, A.~Huang, and H.~Shan, ``Delay and reliability tradeoffs
  in heterogeneous cellular networks,'' \emph{IEEE Transactions on Wireless
  Communications}, pp. 1101--1113, 2016.

\bibitem{ge2015user}
X.~Ge, J.~Ye, Y.~Yang, and Q.~Li, ``User mobility evaluation for 5{G} small
  cell networks based on individual mobility model,'' \emph{IEEE Journal on
  Selected Areas in Communications}, vol.~34, no.~3, pp. 528--541, 2016.

\bibitem{ash1986generalized}
P.~F. Ash and E.~D. Bolker, ``Generalized {D}irichlet tessellations,''
  \emph{Geometriae Dedicata}, vol.~20, no.~2, pp. 209--243, 1986.

\bibitem{mahmoodiusing}
T.~Mahmoodi and S.~Seetharaman, ``On using a {SDN}-based control plane in 5{G}
  mobile networks,'' \emph{Wireless World Research Forum}, vol.~32.

\bibitem{haenggi2009interference}
M.~Haenggi and R.~K. Ganti, \emph{Interference in large wireless
  networks}.\hskip 1em plus 0.5em minus 0.4em\relax Now Publishers Inc, 2009.

\bibitem{spatial_martin}
R.~K. Ganti and M.~Haenggi, ``{Spatial and Temporal Correlation of the
  Interference in {ALOHA} Ad Hoc Networks},'' \emph{IEEE Communications
  Letters}, vol.~13, no.~9, pp. 631--633, Sep. 2009.

\bibitem{dhillon2012modeling}
H.~S. Dhillon, R.~K. Ganti, F.~Baccelli, and J.~G. Andrews, ``Modeling and
  analysis of k-tier downlink heterogeneous cellular networks,'' \emph{IEEE
  Journal on Selected Areas in Communications}, vol.~30, no.~3, pp. 550--560,
  April 2012.

\end{thebibliography}
\end{document}